\documentclass[a4paper]{article}%
\usepackage{amsmath}
\usepackage{amsfonts}
\usepackage{amssymb}
\usepackage{graphicx}%
\providecommand{\U}[1]{\protect\rule{.1in}{.1in}}
\newtheorem{theorem}{Theorem}

\newtheorem{definition}[theorem]{Definition}

\newtheorem{remark}[theorem]{Remark}

\newenvironment{proof}[1][Proof]{\noindent \textbf{#1.} }{\  \rule{0.5em}{0.5em}}
\begin{document}

\title{\textbf{Pseudo minimum phi-divergence estimator for multinomial logistic
regression with complex sample design}}
\author{Elena Castilla, Nirian Mart\'{\i}n and Leandro Pardo\\{\small Department of Statistics and Operations Research, Complutense
University of Madrid, Spain}}
\date{\today}
\maketitle

\begin{abstract}
This article develops the theoretical framework needed to study the
multinomial logistic regression model for complex sample design with pseudo
minimum phi-divergence estimators. Through a numerical example and simulation
study new estimators are proposed for the parameter of the logistic regression
model with overdispersed multinomial distributions for the response variables,
the pseudo minimum Cressie-Read divergence estimators, as well as new
estimators for the intra-cluster correlation coefficient. The results show
that the Binder's method for the intra-cluster correlation coefficient
exhibits an excellent performance when the pseudo minimum Cressie-Read
divergence estimator, with $\lambda=\frac{2}{3}$, is plugged.

\end{abstract}

\bigskip\bigskip



\section{Introduction\label{sec1}}

Multinomial logistic regression is frequently the method of choice when the
response is a qualitative variable, with two or more mutually exclusive
unordered response categories, and interest is in the relationship between the
response variables with respect to their corresponding explanatory variables
or covariates. The $k$ explanatory variables of interest, $\boldsymbol{x}%
=\left(  x_{1},...,x_{k}\right)  ^{T}$, may be binary, categorical, ordinal or
continuos. The multinomial logistic regression procedure is based on assuming
that the $(d+1)$-dimensional response random variable $\boldsymbol{Y}%
=(Y_{1},...,Y_{d+1})^{T}$ is a multinomial random variable of a unique
observation with parameters $\pi_{1}\left(  \boldsymbol{\beta}\right)
,...,\pi_{d+1}\left(  \boldsymbol{\beta}\right)  $ being
\begin{equation}
\pi_{r}\left(  \boldsymbol{\beta}\right)  =\Pr\left(  Y_{r}=1|\boldsymbol{x}%
\right)  =\left\{
\begin{array}
[c]{ll}%
\dfrac{\exp\{\boldsymbol{x}^{T}\boldsymbol{\beta}_{r}\}}{1+%
{\textstyle\sum_{s=1}^{d}}
\exp\{\boldsymbol{x}^{T}\boldsymbol{\beta}_{s}\}}, & r=1,...,d\\
\dfrac{1}{1+%
{\textstyle\sum_{s=1}^{d}}
\exp\{\boldsymbol{x}^{T}\boldsymbol{\beta}_{s}\}}, & r=d+1
\end{array}
\right.  , \label{1.1}%
\end{equation}
with $\boldsymbol{\beta}=(\boldsymbol{\beta}_{1}^{T},...,\boldsymbol{\beta
}_{d}^{T})^{T}$, where $\boldsymbol{\beta}_{r}=\left(  \beta_{1r}%
,...,\beta_{kr}\right)  ^{T}$ is a $k$-dimensional real value vector of
unknown parameters for $r=1,...,d$. An observation of $\boldsymbol{Y}$,
$\boldsymbol{y}$, is any $(d+1)$-dimensional vector with $d$ zeros and a
unique one (classification vector), which is observed together with
explanatory variables $\boldsymbol{x}$. In order to make inferences about
$\boldsymbol{\beta}_{r}$, $r=1,...,d$, a random sample $\left(  \boldsymbol{Y}%
_{i},\boldsymbol{x}_{i}\right)  $, $i=1,...,n$ is considered, where
$\boldsymbol{Y}_{i}=(Y_{i1},...,Y_{i,d+1})^{T}$ and $\boldsymbol{x}%
_{i}=\left(  x_{i1},...,x_{ik}\right)  ^{T}$. For more details about
multinomial logistic regression models see for instance Agresti (2002),
Amemiya (1981), Anderson (1972, 1982, 1984), Engel (1988), Lesaffre (1986),
Lesaffre and Albert (1986, 1989), Liu and Agresti (2005), Mantel (1966), Theil
(1969), McCullagh (1980). In that papers the inferences about the parameters
are carried out on the basis of the maximum likelihood estimator in the case
of the estimation and on the likelihood ratio test and Wald tests in the case
of testing. In Gupta et al. (2006a, 2006b, 2007, 2008) new procedures for
making statistical inference in the multinomial logistic regression were
presented based on phi-divergences measures.

When the data have been collected not under the assumptions of simple random
sampling but in a complex survey, with stratification, clustering, or unequal
selection probabilities, for example, the estimation of the multinomial
logistic regression coefficients and their estimated variances that ignore
these features may be misleading. Discussions of multinomial logistic
regression in sample surveys can be seen in Binder (1983), Roberts, Rao and
Kumar (1987), Skinner, Holt and Smith (1989), Morel (1989), Lehtonen and
Pahkinen (1995) and Morel and Neerchal (2012).

In this paper, we consider the multinomial logistic regression model with
complex survey and we shall introduce for this model the pseudo minimum
phi-divergence estimator for the regressions coefficients, deriving its
asymptotic distribution. As a particular case, we shall obtain the asymptotic
distribution of the pseudo maximum likelihood estimator. In Section
\ref{sec2}, we present some notation as well as some results in relation to
the maximum likelihood estimator. Section \ref{sec3} is devoted to introduce
the pseudo minimum phi-divergence estimator as an extension of the maximum
likelihood estimator as well as its asymptotic distribution. In Section
\ref{sec4} and \ref{sec5}, the numerical example and simulation study are
swown. Finally, in Section \ref{sec6}, some concluding remarks are given.

\section{Multinomial logistic regression model for complex sample
design\label{sec2}}

We shall assume that the population under consideration is divided into $H$
distinct strata. In each stratum $h$, the sample is consisted of $n_{h}$
clusters, $h=1,...,H$, and each cluster is comprised of $m_{hi}$ units,
$h=1,...,H,$ $i=1,...,n_{h}$. Let
\begin{equation}
\boldsymbol{y}_{hij}=\left(  y_{hij1},....,y_{hij,d+1}\right)  ^{T},\text{
}h=1,...,H,\text{ }i=1,...,n_{h},\text{ }j=1,...,m_{hi} \label{2.1}%
\end{equation}
be the $(d+1)$-dimensional classification vectors, with $y_{hijr}$ $=1$ and
$y_{hijs}$ $=0$ for $s\in\{1,...,d+1\}-\{r\}$ if the $j$-th unit selected from
the $i$-th cluster of the $h$-th stratum fall in the $r$-th category. Let
$\boldsymbol{x}_{hij}=\left(  x_{hij1},....,x_{hijk}\right)  ^{T}$ be a
$k$-dimensional vector of explanatory variables associated with the $i$-th
cluster in the $h$-th stratum for the $j$-th individual. We shall also denote
by $w_{hi}$ the sampling weight from the $i$-th cluster of the $h$-th stratum.
For each $i$, $h$ and $j$, the expectation of the $r$-th element of
$\boldsymbol{Y}_{hij}=(Y_{hij1},...,Y_{hij,d+1})^{T}$, with a realization
$\boldsymbol{y}_{hij}$, is determined by the multinomial logistic regression
relationship%
\begin{equation}
\pi_{hijr}\left(  \boldsymbol{\beta}\right)  =\left\{
\begin{array}
[c]{ll}%
\dfrac{\exp\{\boldsymbol{x}_{hij}^{T}\boldsymbol{\beta}_{r}\}}{1+%
{\textstyle\sum_{s=1}^{d}}
\exp\{\boldsymbol{x}_{hij}^{T}\boldsymbol{\beta}_{s}\}}, & r=1,...,d\\
\dfrac{1}{1+%
{\textstyle\sum_{s=1}^{d}}
\exp\{\boldsymbol{x}_{hij}^{T}\boldsymbol{\beta}_{s}\}}, & r=d+1
\end{array}
\right.  , \label{2.1.0}%
\end{equation}
with $\boldsymbol{\beta}_{r}=\left(  \beta_{1r},...,\beta_{kr}\right)  ^{T}\in%
\mathbb{R}
^{k}$, $r=1,...,d$. We shall denote by $\boldsymbol{\pi}_{hij}\left(
\boldsymbol{\beta}\right)  $ the $(d+1)$-dimensional probability vector
\begin{equation}
\boldsymbol{\pi}_{hij}\left(  \boldsymbol{\beta}\right)  =\left(  \pi
_{hij1}\left(  \boldsymbol{\beta}\right)  ,...,\pi_{hij,d+1}\left(
\boldsymbol{\beta}\right)  \right)  ^{T}. \label{2.2}%
\end{equation}

The parameter space associated to the multinomial logistic regression model
considered in (\ref{2.1.0}) is given by
\[
\Theta=\{\boldsymbol{\beta}=(\boldsymbol{\beta}_{1}^{T},...,\boldsymbol{\beta
}_{d}^{T})^{T},\text{ }\boldsymbol{\beta}_{j}=\left(  \beta_{j1}%
,...,\beta_{jk}\right)  ^{T}\in\mathbb{R}^{k},\text{ }j=1,...,d\}=\mathbb{R}%
^{dk}.
\]
In this context and taking into account the weights $w_{hi}$, the pseudo
log-likelihood, $\mathcal{L}\left(  \boldsymbol{\beta}\right)  $, for the
multinomial logistic regression model given in (\ref{2.1.0}) has the
expression
\begin{equation}
\mathcal{L}\left(  \boldsymbol{\beta}\right)  =%
{\displaystyle\sum\limits_{h=1}^{H}}
{\displaystyle\sum\limits_{i=1}^{n_{h}}}
{\displaystyle\sum\limits_{j=1}^{m_{hi}}}
w_{hi}\log\boldsymbol{\pi}_{hij}^{T}\left(  \boldsymbol{\beta}\right)
\boldsymbol{y}_{hij}, \label{2.3}%
\end{equation}
where $\log\boldsymbol{\pi}_{hij}\left(  \boldsymbol{\beta}\right)  =\left(
\log\pi_{hij1}\left(  \boldsymbol{\beta}\right)  ,...,\log\pi_{hij,d+1}\left(
\boldsymbol{\beta}\right)  \right)  ^{T}$. For more details about
$\mathcal{L}\left(  \boldsymbol{\beta}\right)  $ see for instance Morel (1989)
and Morel and Neerchal (2012).

In practice, it is not a strong assumption to consider that the expectation of
the $r$-th component of $\boldsymbol{Y}_{hij}$ does not depend on $j$, i.e.,%
\[
\pi_{hijr}\left(  \boldsymbol{\beta}\right)  =\pi_{hir}\left(
\boldsymbol{\beta}\right)  ,\quad j=1,...,m_{hi},
\]
where $\pi_{hijr}\left(  \boldsymbol{\beta}\right)  =\mathrm{E}[Y_{hijr}%
]=\Pr(Y_{hijr}=1)$. This is related to a common vector of explanatory
variables $\boldsymbol{x}_{hi}=\left(  x_{hi1},....,x_{hik}\right)  ^{T}$ for
all the individuals in the $i$-th cluster of the $h$-th stratum and we shall
denote $\boldsymbol{\pi}_{hi}\left(  \boldsymbol{\beta}\right)  $ instead of
$\boldsymbol{\pi}_{hij}\left(  \boldsymbol{\beta}\right)  $ the vector mean
associated to $\boldsymbol{Y}_{hij}$. Let%
\begin{equation}
\widehat{\boldsymbol{Y}}_{hi}=%
{\displaystyle\sum\limits_{j=1}^{m_{hi}}}
\boldsymbol{Y}_{hij}=\left(
{\displaystyle\sum\limits_{j=1}^{m_{hi}}}
Y_{hij1},...,%
{\displaystyle\sum\limits_{j=1}^{m_{hi}}}
Y_{hij,d+1}\right)  ^{T}=(\widehat{Y}_{hi1},...,\widehat{Y}_{hi,d+1})^{T}
\label{2.201}%
\end{equation}
be the random vector of counts in the $i$-th cluster of the $h$-th stratum.
Under homogeneity assumption within the clusters, the pseudo log-likelihood is%
\begin{align}
\mathcal{L}\left(  \boldsymbol{\beta}\right)   &  =%
{\displaystyle\sum\limits_{h=1}^{H}}
{\displaystyle\sum\limits_{i=1}^{n_{h}}}
{\displaystyle\sum\limits_{j=1}^{m_{hi}}}
w_{hi}\log\boldsymbol{\pi}_{hi}^{T}\left(  \boldsymbol{\beta}\right)
\boldsymbol{y}_{hij}\nonumber\\
&  =%
{\displaystyle\sum\limits_{h=1}^{H}}
{\displaystyle\sum\limits_{i=1}^{n_{h}}}
w_{hi}\log\boldsymbol{\pi}_{hi}^{T}\left(  \boldsymbol{\beta}\right)
\widehat{\boldsymbol{y}}_{hi}. \label{2.3.1}%
\end{align}

The pseudo maximum likelihood estimator $\widehat{\boldsymbol{\beta}}_{P}$ of
$\boldsymbol{\beta}$ is obtained maximizing in $\boldsymbol{\beta}$ the pseudo
log-likelihood given in (\ref{2.3.1}). This estimator can be obtained as the
solution of the system of equations%
\begin{equation}%
{\displaystyle\sum\limits_{h=1}^{H}}
{\displaystyle\sum\limits_{i=1}^{n_{h}}}
w_{hi}\frac{\partial\boldsymbol{\pi}_{hi}^{\ast T}\left(  \boldsymbol{\beta
}\right)  }{\partial\boldsymbol{\beta}}\boldsymbol{\Delta}^{-1}%
(\boldsymbol{\pi}_{hi}^{\ast}\left(  \boldsymbol{\beta}\right)
)\boldsymbol{r}_{hi}^{\ast}\left(  \boldsymbol{\beta}\right)  =\boldsymbol{0}%
_{dk}, \label{2.4}%
\end{equation}
being
\begin{align*}
\frac{\partial\boldsymbol{\pi}_{hi}^{\ast T}\left(  \boldsymbol{\beta}\right)
}{\partial\boldsymbol{\beta}}  &  =\mathcal{\boldsymbol{\Delta}}%
(\boldsymbol{\pi}_{hi}^{\ast}\left(  \boldsymbol{\beta}\right)  )\otimes
\boldsymbol{x}_{hi},\\
\boldsymbol{\Delta}(\boldsymbol{\pi}_{hi}^{\ast}\left(  \boldsymbol{\beta
}\right)  )  &  =\mathrm{diag}(\boldsymbol{\pi}_{hi}^{\ast}\left(
\boldsymbol{\beta}\right)  )-\boldsymbol{\pi}_{hi}^{\ast}\left(
\boldsymbol{\beta}\right)  \boldsymbol{\pi}_{hi}^{\ast T}\left(
\boldsymbol{\beta}\right)  ,\\
\boldsymbol{r}_{hi}^{\ast}\left(  \boldsymbol{\beta}\right)   &
=\widehat{\boldsymbol{y}}_{hi}^{\ast}-m_{hi}\boldsymbol{\pi}_{hi}^{\ast
}\left(  \boldsymbol{\beta}\right)  .
\end{align*}
With superscript $^{\ast}$ on a vector we denote the vector obtained
deleting\ the last component from the initial vector, and thus
$\boldsymbol{\pi}_{hi}^{\ast}\left(  \boldsymbol{\beta}\right)  =\left(
\pi_{hi1}\left(  \boldsymbol{\beta}\right)  ,...,\pi_{hid}\left(
\boldsymbol{\beta}\right)  \right)  ^{T}$ and $\widehat{\boldsymbol{y}}%
_{hi}^{\ast}=\left(  \widehat{y}_{hi1}^{\ast},...,\widehat{y}_{hid}^{\ast
}\right)  ^{T}$. The system of equations (\ref{2.4}) can be written as
$\boldsymbol{u}\left(  \boldsymbol{\beta}\right)  =\boldsymbol{0}_{dk}$,
being
\begin{align}
\boldsymbol{u}\left(  \boldsymbol{\beta}\right)   &  =%
{\displaystyle\sum\limits_{h=1}^{H}}
{\displaystyle\sum\limits_{i=1}^{n_{h}}}
\boldsymbol{u}_{hi}\left(  \boldsymbol{\beta}\right)  ,\label{Un}\\
\boldsymbol{u}_{hi}\left(  \boldsymbol{\beta}\right)   &  =w_{hi}%
\boldsymbol{r}_{hi}^{\ast}\left(  \boldsymbol{\beta}\right)  \otimes
\boldsymbol{x}_{hi}. \label{Un2}%
\end{align}

\section{Pseudo minimum phi-divergence estimator: asymptotic
distribution\label{sec3}}

In this Section we shall introduce, for the fist time, the pseudo minimum
phi-divergence estimator, $\widehat{\boldsymbol{\beta}}_{\phi,P}$, of the
parameter $\boldsymbol{\beta}$ as a natural extension of the pseudo maximum
likelihood estimator $\widehat{\boldsymbol{\beta}}_{P}$. We define the
following theoretical probability vector%
\[
\boldsymbol{\pi}\left(  \boldsymbol{\beta}\right)  =\frac{1}{\tau}%
(w_{11}m_{11}\boldsymbol{\pi}_{11}^{T}(\boldsymbol{\beta}),...,w_{1n_{1}%
}m_{1n_{1}}\boldsymbol{\pi}_{1n_{1}}^{T}(\boldsymbol{\beta}),...,w_{H1}%
m_{H1}\boldsymbol{\pi}_{H1}^{T}\left(  \boldsymbol{\beta}\right)
,...,w_{Hn_{H}}m_{Hn_{H}}\boldsymbol{\pi}_{Hn_{H}}^{T}(\boldsymbol{\beta
}))^{T},
\]
with
\begin{equation}
\tau=%
{\displaystyle\sum\limits_{h=1}^{H}}
{\displaystyle\sum\limits_{i=1}^{n_{h}}}
w_{hi}m_{hi} \label{2.202}%
\end{equation}
being a known value. Based on $\widehat{\boldsymbol{y}}_{hi}$, observation of
$\widehat{\boldsymbol{Y}}_{hi}$\ defined in (\ref{2.201}), we consider the
vector $\widehat{\boldsymbol{y}}_{h}$ for each stratum $h$,%
\[
\widehat{\boldsymbol{y}}_{h}=(w_{h1}\widehat{\boldsymbol{y}}_{h1}%
^{T},...,w_{hn_{h}}\widehat{\boldsymbol{y}}_{hn_{h}}^{T})^{T}.
\]
We shall also consider the non-parametric probability vector
\begin{align*}
\widehat{\boldsymbol{p}}  &  =\frac{1}{\tau}(\widehat{\boldsymbol{y}}_{1}%
^{T},...,\widehat{\boldsymbol{y}}_{H}^{T})^{T}\\
&  =\frac{1}{\tau}(w_{11}\widehat{\boldsymbol{y}}_{11}^{T},...,w_{1n_{1}%
}\widehat{\boldsymbol{y}}_{1n_{1}}^{T},...,w_{H1}\widehat{\boldsymbol{y}}%
_{H1}^{T},...,w_{Hn_{H}}\widehat{\boldsymbol{y}}_{Hn_{H}}^{T})^{T}.
\end{align*}

The Kullback-Leibler divergence between the probability vectors
$\widehat{\boldsymbol{p}}$ and $\boldsymbol{\pi}\left(  \boldsymbol{\beta
}\right)  $ is given by
\begin{align}
d_{K\mathrm{-}L}\left(  \widehat{\boldsymbol{p}},\boldsymbol{\pi}\left(
\boldsymbol{\beta}\right)  \right)   &  =\frac{1}{\tau}%
{\displaystyle\sum\limits_{h=1}^{H}}
{\displaystyle\sum\limits_{i=1}^{n_{h}}}
w_{hi}%
{\displaystyle\sum\limits_{s=1}^{d+1}}
\widehat{y}_{his}\log\frac{\widehat{y}_{his}}{m_{hi}\pi_{his}\left(
\boldsymbol{\beta}\right)  }\label{3.1}\\
&  =K-\frac{1}{\tau}%
{\displaystyle\sum\limits_{h=1}^{H}}
{\displaystyle\sum\limits_{i=1}^{n_{h}}}
w_{hi}%
{\displaystyle\sum\limits_{s=1}^{d+1}}
\widehat{y}_{his}\log\pi_{his}\left(  \boldsymbol{\beta}\right) \nonumber\\
&  =K-\frac{1}{\tau}%
{\displaystyle\sum\limits_{h=1}^{H}}
{\displaystyle\sum\limits_{i=1}^{n_{h}}}
w_{hi}\log\boldsymbol{\pi}_{hi}^{T}\left(  \boldsymbol{\beta}\right)
\widehat{\boldsymbol{y}}_{hi},\nonumber
\end{align}
with $K$ being a constant not dependent of $\boldsymbol{\beta}$. Based on
(\ref{2.3.1}) and (\ref{3.1}), we can define the pseudo maximum likelihood
estimator for the multinomial logistic regression model given in (\ref{2.1.0})
by
\begin{equation}
\widehat{\boldsymbol{\beta}}_{P}=\arg\min_{\boldsymbol{\beta\in\Theta}%
}d_{K\mathrm{-}L}\left(  \widehat{\boldsymbol{p}},\boldsymbol{\pi}\left(
\boldsymbol{\beta}\right)  \right)  . \label{3.2}%
\end{equation}
But Kullback-Leibler divergence is a particular divergence measure in the
family of phi-divergence measures,
\begin{equation}
d_{\phi}\left(  \widehat{\boldsymbol{p}},\boldsymbol{\pi}\left(
\boldsymbol{\beta}\right)  \right)  =\frac{1}{\tau}\sum\limits_{h=1}^{H}%
\sum\limits_{i=1}^{n_{h}}w_{hi}m_{hi}\sum\limits_{s=1}^{d+1}\pi_{his}\left(
\boldsymbol{\beta}\right)  \phi\left(  \frac{\widehat{y}_{his}}{m_{hi}%
\pi_{his}\left(  \boldsymbol{\beta}\right)  }\right)  , \label{3.3}%
\end{equation}
where $\phi\in\Phi^{\ast}$ is the class of all convex functions $\phi\left(
x\right)  $, defined for $x>0$, such that at $x=1$, $\phi\left(  1\right)
=0$, $\phi^{\prime\prime}\left(  1\right)  >0,$ and at $x=0$, $0\phi\left(
0/0\right)  =0$ and $0\phi\left(  p/0\right)  =\lim_{u\rightarrow\infty}%
\phi\left(  u\right)  /u$. For every $\phi\in\Phi^{\ast}$ differentiable at
$x=1$, the function
\[
\varphi\left(  x\right)  \equiv\phi\left(  x\right)  -\phi^{\prime}\left(
1\right)  \left(  x-1\right)
\]
also belongs to $\Phi^{\ast}$. Then we have $d_{\varphi}\left(
\widehat{\boldsymbol{p}},\boldsymbol{\pi}\left(  \boldsymbol{\beta}\right)
\right)  =d_{\phi}\left(  \widehat{\boldsymbol{p}},\boldsymbol{\pi}\left(
\boldsymbol{\beta}\right)  \right)  $, and $\varphi$ has the additional
property that $\varphi^{\prime}\left(  1\right)  =0$. Because the two
divergence measures are equivalent, we can consider the set $\Phi^{\ast}$ to
be equivalent to the set
\[
\Phi\equiv\Phi^{\ast}\cap\left\{  \phi:\phi^{\prime}\left(  1\right)
=0\right\}  .
\]
In what follows, we give our theoretical results for $\phi\in\Phi$, but we
often apply them to choices of functions in $\Phi^{\ast}$.

An equivalent definition of (\ref{3.3}) is a weighted version of
phi-divergences between the cluster non-parametric probabilities and
theoretical probabilities%
\[
d_{\phi}\left(  \widehat{\boldsymbol{p}},\boldsymbol{\pi}\left(
\boldsymbol{\beta}\right)  \right)  =\sum\limits_{h=1}^{H}\sum\limits_{i=1}%
^{n_{h}}\frac{w_{hi}m_{hi}}{\tau}d_{\phi}\left(  \tfrac
{\widehat{\boldsymbol{y}}_{hi}}{m_{hi}},\boldsymbol{\pi}_{hi}%
(\boldsymbol{\beta})\right)  ,
\]
where%
\[
d_{\phi}\left(  \tfrac{\widehat{\boldsymbol{y}}_{hi}}{m_{hi}},\boldsymbol{\pi
}_{hi}(\boldsymbol{\beta})\right)  =\sum\limits_{s=1}^{d+1}\pi_{his}\left(
\boldsymbol{\beta}\right)  \phi\left(  \frac{\widehat{y}_{his}}{m_{hi}%
\pi_{his}\left(  \boldsymbol{\beta}\right)  }\right)  .
\]
For more details about phi-divergences measures see Pardo (2005).

Based on (\ref{3.2}) and (\ref{3.3}) we shall introduce, in this paper, the
pseudo minimum phi-divergence estimator for the parameter $\boldsymbol{\beta}$
in the multinomial logistic regression model under complex survey defined in
(\ref{2.1.0}) as follows,

\begin{definition}
We consider the multinomial logistic regression model with complex survey
defined in (\ref{2.1.0}). The pseudo minimum phi-divergence estimator of
$\boldsymbol{\beta}$ is defined as%
\[
\widehat{\boldsymbol{\beta}}_{\phi,P}=\arg\min_{\boldsymbol{\beta}\in\Theta
}d_{\phi}\left(  \widehat{\boldsymbol{p}},\boldsymbol{\pi}\left(
\boldsymbol{\beta}\right)  \right)  ,
\]
where $d_{\phi}\left(  \widehat{\boldsymbol{p}},\boldsymbol{\pi}\left(
\boldsymbol{\beta}\right)  \right)  $, the phi-divergence measure between the
probability vectors $\widehat{\boldsymbol{p}}$ and $\boldsymbol{\pi}\left(
\boldsymbol{\beta}\right)  $, is given in (\ref{3.3}).
\end{definition}

For $\phi(x)=x\log x-x+1$ the associated phi-divergence (\ref{3.3}) coincides
with the Kullback-Leibler divergence (\ref{3.1}), therefore the pseudo minimum
phi-divergence estimator of $\boldsymbol{\beta}$ based on $\phi(x)$ contains
as special case the pseudo maximum likelihood estimator. With the same
philosophy, the following result generalizes $\boldsymbol{u}_{hi}\left(
\boldsymbol{\beta}\right)  $ given in (\ref{Un2}) and later this result plays
an important role for the asymptotic distribution of the pseudo minimum
phi-divergence estimator, $\widehat{\boldsymbol{\beta}}_{\phi,P}$.

\begin{theorem}
\label{Th0}The pseudo minimum phi-divergence estimator of $\boldsymbol{\beta}%
$, $\widehat{\boldsymbol{\beta}}_{\phi,P}$, is obtained by solving the system
of equations $\boldsymbol{u}_{\phi}\left(  \boldsymbol{\beta}\right)
=\boldsymbol{0}_{dk}$, where%
\begin{align}
\boldsymbol{u}_{\phi}\left(  \boldsymbol{\beta}\right)   &  =\sum
\limits_{h=1}^{H}\sum\limits_{i=1}^{n_{h}}\boldsymbol{u}_{\phi,hi}\left(
\boldsymbol{\beta}\right)  ,\label{3.09}\\
\boldsymbol{u}_{\phi,hi}\left(  \boldsymbol{\beta}\right)   &  =\frac
{w_{hi}m_{hi}}{\phi^{\prime\prime}(1)}\boldsymbol{\Delta}(\boldsymbol{\pi
}_{hi}^{\ast}\left(  \boldsymbol{\beta}\right)  )\boldsymbol{f}_{\phi
,hi}^{\ast}(\tfrac{\widehat{\boldsymbol{y}}_{hi}}{m_{hi}},\boldsymbol{\beta
})\otimes\boldsymbol{x}_{hi}, \label{3.010}%
\end{align}
where
\begin{align}
\boldsymbol{f}_{\phi,hi}^{\ast}(\tfrac{\widehat{\boldsymbol{y}}_{hi}}{m_{hi}%
},\boldsymbol{\beta})  &  =(f_{\phi,hi1}(\tfrac{\widehat{y}_{hi1}}{m_{hi}%
},\boldsymbol{\beta}),...,f_{\phi,hid}(\tfrac{\widehat{y}_{hid}}{m_{hi}%
},\boldsymbol{\beta}))^{T},\nonumber\\
f_{\phi,his}(x,\boldsymbol{\beta})  &  =\frac{x}{\pi_{his}(\boldsymbol{\beta
})}\phi^{\prime}\left(  \frac{x}{\pi_{his}(\boldsymbol{\beta})}\right)
-\phi\left(  \frac{x}{\pi_{his}(\boldsymbol{\beta})}\right)  \label{3.010b}%
\end{align}

\end{theorem}

\begin{proof}
The pseudo minimum phi-divergence estimator of $\boldsymbol{\beta}$,
$\widehat{\boldsymbol{\beta}}_{\phi,P}$, is obtained by solving the system of
equations $\frac{\partial}{\partial\boldsymbol{\beta}}d_{\phi}\left(
\widehat{\boldsymbol{p}},\boldsymbol{\pi}\left(  \boldsymbol{\beta}\right)
\right)  =\boldsymbol{0}_{dk}$, and then it is also obtained from
$\boldsymbol{u}_{\phi}\left(  \boldsymbol{\beta}\right)  =\boldsymbol{0}_{dk}%
$, where%
\[
\boldsymbol{u}_{\phi}\left(  \boldsymbol{\beta}\right)  =-\frac{\tau}%
{\phi^{\prime\prime}(1)}\frac{\partial}{\partial\boldsymbol{\beta}}d_{\phi
}\left(  \widehat{\boldsymbol{p}},\boldsymbol{\pi}\left(  \boldsymbol{\beta
}\right)  \right)  =\sum\limits_{h=1}^{H}\sum\limits_{i=1}^{n_{h}%
}\boldsymbol{u}_{\phi,hi}\left(  \boldsymbol{\beta}\right)  ,
\]
with%
\begin{align}
\boldsymbol{u}_{\phi,hi}\left(  \boldsymbol{\beta}\right)   &  =-\frac
{w_{hi}m_{hi}}{\phi^{\prime\prime}(1)}\frac{\partial}{\partial
\boldsymbol{\beta}}d_{\phi}\left(  \tfrac{\widehat{\boldsymbol{y}}_{hi}%
}{m_{hi}},\boldsymbol{\pi}_{hi}(\boldsymbol{\beta})\right)  =\frac
{w_{hi}m_{hi}}{\phi^{\prime\prime}(1)}%
{\displaystyle\sum\limits_{s=1}^{d+1}}
\frac{\partial\pi_{his}(\boldsymbol{\beta})}{\partial\boldsymbol{\beta}%
}f_{\phi,his}(\tfrac{\widehat{y}_{his}}{m_{hi}},\boldsymbol{\beta})\nonumber\\
&  =\frac{w_{hi}m_{hi}}{\phi^{\prime\prime}(1)}\frac{\partial\boldsymbol{\pi
}_{hi}^{T}(\boldsymbol{\beta})}{\partial\boldsymbol{\beta}}\boldsymbol{f}%
_{\phi,hi}(\tfrac{\widehat{\boldsymbol{y}}_{hi}}{m_{hi}},\boldsymbol{\beta}),
\label{uu}%
\end{align}
and%
\[
\boldsymbol{f}_{\phi,hi}(\tfrac{\widehat{\boldsymbol{y}}_{hi}}{m_{hi}%
},\boldsymbol{\beta})=(f_{\phi,hi1}(\tfrac{\widehat{y}_{hi1}}{m_{hi}%
},\boldsymbol{\beta}),...,f_{\phi,hi,d+1}(\tfrac{\widehat{y}_{hi,d+1}}{m_{hi}%
},\boldsymbol{\beta}))^{T}.
\]
Since
\begin{equation}
\frac{\partial\boldsymbol{\pi}_{hi}^{T}(\boldsymbol{\beta})}{\partial
\boldsymbol{\beta}}=\left(  \boldsymbol{I}_{d\times d},\boldsymbol{0}%
_{d\times1}\right)  \boldsymbol{\Delta}(\boldsymbol{\pi}_{hi}\left(
\boldsymbol{\beta}\right)  )\otimes\boldsymbol{x}_{hi}, \label{der}%
\end{equation}
the expression of $\boldsymbol{u}_{\phi,hi}\left(  \boldsymbol{\beta}\right)
$\ is rewritten as (\ref{3.010}).
\end{proof}

\begin{remark}
An important family of divergence measures is obtained by restricting $\phi$
from the family of convex\ functions to the Cressie-Read subfamily%
\begin{equation}
\phi_{\lambda}(x)=\left\{
\begin{array}
[c]{ll}%
\frac{1}{\lambda(1+\lambda)}\left[  x^{\lambda+1}-x-\lambda(x-1)\right]  , &
\lambda\in%
\mathbb{R}
-\{-1,0\}\\
\lim_{\upsilon\rightarrow\lambda}\frac{1}{\upsilon(1+\upsilon)}\left[
x^{\upsilon+1}-x-\upsilon(x-1)\right]  , & \lambda\in\{-1,0\}
\end{array}
\right.  . \label{CR}%
\end{equation}
We can observe that for $\lambda=0$, we have
\[
\phi_{\lambda=0}(x)=\lim_{\upsilon\rightarrow0}\frac{1}{\upsilon(1+\upsilon
)}\left[  x^{\upsilon+1}-x-\upsilon(x-1)\right]  =x\log x-x+1,
\]
and the associated phi-divergence (\ref{3.3}), coincides with the
Kullback-Leibler divergence (\ref{3.1}), therefore the pseudo minimum
phi-divergence estimator of $\boldsymbol{\beta}$ based on $\phi_{\lambda}(x)$
contains as special case the pseudo maximum likelihood estimator and
$\boldsymbol{u}_{hi}\left(  \boldsymbol{\beta}\right)  $ given in (\ref{Un2})
matches $\boldsymbol{u}_{\phi,hi}\left(  \boldsymbol{\beta}\right)  $ given in
(\ref{3.010}). For the Cressie-Read subfamily it is established that for
$\lambda\neq-1$, $\boldsymbol{u}_{\phi_{\lambda}}\left(  \boldsymbol{\beta
}\right)  =%
{\textstyle\sum\nolimits_{h=1}^{H}}
{\textstyle\sum\nolimits_{i=1}^{n_{i}}}
\boldsymbol{u}_{\phi_{\lambda},hi}\left(  \boldsymbol{\beta}\right)  $, where%
\[
\boldsymbol{u}_{\phi_{\lambda},hi}\left(  \boldsymbol{\beta}\right)
=\frac{w_{hi}}{(\lambda+1)m_{hi}^{\lambda}}\frac{\partial\boldsymbol{\pi}%
_{hi}^{T}(\boldsymbol{\beta})}{\partial\boldsymbol{\beta}}\mathrm{diag}%
^{-(\lambda+1)}(\boldsymbol{\pi}_{hi}(\boldsymbol{\beta}%
))\widehat{\boldsymbol{y}}_{hi}^{\lambda+1},
\]
since we have (\ref{uu}) with
\begin{equation}
\boldsymbol{f}_{\phi_{\lambda},hi}(\tfrac{\widehat{\boldsymbol{y}}_{hi}%
}{m_{hi}},\boldsymbol{\beta})=\frac{1}{\lambda+1}\left(  \frac{1}%
{m_{hi}^{\lambda+1}}\mathrm{diag}^{-(\lambda+1)}(\boldsymbol{\pi}%
_{hi}(\boldsymbol{\beta}))\widehat{\boldsymbol{y}}_{hi}^{\lambda
+1}-\boldsymbol{1}_{d+1}\right)  , \label{f}%
\end{equation}
From (\ref{der}) and
\begin{align*}
\boldsymbol{\Delta}(\boldsymbol{\pi}_{hi}\left(  \boldsymbol{\beta}\right)
)\mathrm{diag}^{-(\lambda+1)}(\boldsymbol{\pi}_{hi}(\boldsymbol{\beta}))  &
=\boldsymbol{\Delta}(\boldsymbol{\pi}_{hi}\left(  \boldsymbol{\beta}\right)
)\mathrm{diag}^{-1}(\boldsymbol{\pi}_{hi}(\boldsymbol{\beta}))\mathrm{diag}%
^{-\lambda}(\boldsymbol{\pi}_{hi}(\boldsymbol{\beta}))\\
&  =\mathrm{diag}^{-\lambda}(\boldsymbol{\pi}_{hi}(\boldsymbol{\beta
}))-\boldsymbol{\pi}_{hi}(\boldsymbol{\beta})\boldsymbol{1}_{d+1}%
^{T}\mathrm{diag}^{-\lambda}(\boldsymbol{\pi}_{hi}(\boldsymbol{\beta})),
\end{align*}
it is concluded that%
\begin{align}
\boldsymbol{u}_{\phi_{\lambda},hi}\left(  \boldsymbol{\beta}\right)   &
=\frac{w_{hi}}{(\lambda+1)m_{hi}^{\lambda}}\left(  \mathrm{diag}^{-\lambda
}(\boldsymbol{\pi}_{hi}^{\ast}(\boldsymbol{\beta}))\widehat{\boldsymbol{y}%
}_{hi}^{\ast,\lambda+1}-[\boldsymbol{1}_{d+1}^{T}\mathrm{diag}^{-\lambda
}(\boldsymbol{\pi}_{hi}(\boldsymbol{\beta}))\widehat{\boldsymbol{y}}%
_{hi}^{\lambda+1}]\boldsymbol{\pi}_{hi}^{\ast}(\boldsymbol{\beta})\right)
\otimes\boldsymbol{x}_{hi}\nonumber\\
&  =\frac{w_{hi}}{(\lambda+1)m_{hi}^{\lambda}}\left\{  \mathrm{diag}^{\lambda
}(\boldsymbol{\epsilon}_{hi}^{\ast})\widehat{\boldsymbol{y}}_{hi}^{\ast
}-\left[  \boldsymbol{1}_{d+1}^{T}\mathrm{diag}^{\lambda}(\boldsymbol{\epsilon
}_{hi})\widehat{\boldsymbol{y}}_{hi}\right]  \boldsymbol{\pi}_{hi}^{\ast
}(\boldsymbol{\beta})\right\}  \otimes\boldsymbol{x}_{hi}, \label{ulam}%
\end{align}
where%
\[
\boldsymbol{\epsilon}_{hi}=\mathrm{diag}^{-1}(\boldsymbol{\pi}_{hi}%
(\boldsymbol{\beta}))\widehat{\boldsymbol{y}}_{hi},\qquad\boldsymbol{\epsilon
}_{hi}^{\ast}=\mathrm{diag}^{-1}(\boldsymbol{\pi}_{hi}^{\ast}%
(\boldsymbol{\beta}))\widehat{\boldsymbol{y}}_{hi}^{\ast}.
\]
Notice that replacing $\lambda=0$ in $\boldsymbol{u}_{\phi_{\lambda}%
,hi}\left(  \boldsymbol{\beta}\right)  $ given in (\ref{ulam}),
$\boldsymbol{u}_{hi}\left(  \boldsymbol{\beta}\right)  $ given in (\ref{Un2})
is obtained. For $\lambda=-1$ in (\ref{f}), we have%
\[
\lim_{\lambda\rightarrow-1}\boldsymbol{f}_{\phi_{\lambda},hi}(\tfrac
{\widehat{\boldsymbol{y}}_{hi}}{m_{hi}},\boldsymbol{\beta})=\log\left(
\mathrm{diag}^{-1}(\boldsymbol{\pi}_{hi}(\boldsymbol{\beta}))\frac
{\widehat{\boldsymbol{y}}_{hi}}{m_{hi}}\right)  ,
\]
and therefore%
\[
\lim_{\lambda\rightarrow-1}\boldsymbol{u}_{\phi_{\lambda},hi}\left(
\boldsymbol{\beta}\right)  =w_{hi}m_{hi}\boldsymbol{\Delta}(\boldsymbol{\pi
}_{hi}^{\ast}\left(  \boldsymbol{\beta}\right)  )\log\left(  \mathrm{diag}%
^{-1}(\boldsymbol{\pi}_{hi}^{\ast}(\boldsymbol{\beta}))\frac
{\widehat{\boldsymbol{y}}_{hi}^{\ast}}{m_{hi}}\right)  \otimes\boldsymbol{x}%
_{hi}.
\]
The family of pseudo minimum divergence estimators, obtained from
$\phi_{\lambda}(x)$given in (\ref{CR}), will be called the pseudo minimum
Cressie-Read divergence estimators and for $\boldsymbol{\beta}$\ they will be
denoted by $\widehat{\boldsymbol{\beta}}_{\phi_{\lambda},P}$. This family of
estimators will be used in Sections \ref{sec4} and \ref{sec5}.
\end{remark}

In the following theorem we shall present the asymptotic distribution of the
pseudo minimum phi-divergence estimator, $\widehat{\boldsymbol{\beta}}%
_{\phi,P}$.

\begin{theorem}
\label{Th1}Let $\widehat{\boldsymbol{\beta}}_{\phi,P}$ the pseudo minimum
phi-divergence estimator of parameter $\boldsymbol{\beta}$ for a multinomial
logistic regression model with complex survey, $n=%
{\displaystyle\sum\limits_{h=1}^{H}}
n_{h}$ the total of clusters in all the strata of the sample and $\eta
_{h}^{\ast}$\ an unknown proportion obtained as $\lim_{n\rightarrow\infty
}\frac{n_{h}}{n}=\eta_{h}^{\ast}$, $h=1,...,H$. Then we have%
\[
\sqrt{n}(\widehat{\boldsymbol{\beta}}_{\phi,P}-\boldsymbol{\beta}%
_{0})\overset{\mathcal{L}}{\underset{n\mathcal{\rightarrow}\infty
}{\longrightarrow}}\mathcal{N}\left(  \boldsymbol{0}_{dk},\mathbf{H}%
^{-1}\left(  \boldsymbol{\beta}_{0}\right)  \mathbf{G}\left(
\boldsymbol{\beta}_{0}\right)  \mathbf{H}^{-1}\left(  \boldsymbol{\beta}%
_{0}\right)  \right)  ,
\]
where
\[
\mathbf{H}\left(  \boldsymbol{\beta}\right)  =\lim_{n\rightarrow\infty
}\mathbf{H}_{n}\left(  \boldsymbol{\beta}\right)  =%
{\displaystyle\sum\limits_{h=1}^{H}}
\eta_{h}^{\ast}\lim_{n_{h}\rightarrow\infty}\mathbf{H}_{n_{h}}^{(h)}\left(
\boldsymbol{\beta}\right)  \text{,\quad}\mathbf{G}\left(  \boldsymbol{\beta
}\right)  =\lim_{n\rightarrow\infty}\mathbf{G}_{n}\left(  \boldsymbol{\beta
}\right)  =%
{\displaystyle\sum\limits_{h=1}^{H}}
\eta_{h}^{\ast}\lim_{n_{h}\rightarrow\infty}\mathbf{G}_{n_{h}}^{(h)}\left(
\boldsymbol{\beta}\right)  ,
\]
with%
\[
\mathbf{H}_{n}\left(  \boldsymbol{\beta}\right)  =\frac{1}{n}%
{\displaystyle\sum\limits_{h=1}^{H}}
{\displaystyle\sum\limits_{i=1}^{n_{h}}}
w_{hi}m_{hi}\boldsymbol{\Delta}(\boldsymbol{\pi}_{hi}^{\ast}\left(
\boldsymbol{\beta}\right)  )\otimes\boldsymbol{x}_{hi}\boldsymbol{x}_{hi}%
^{T},\text{\quad}\mathbf{H}_{n_{h}}^{(h)}\left(  \boldsymbol{\beta}\right)
=\frac{1}{n_{h}}%
{\displaystyle\sum\limits_{i=1}^{n_{h}}}
w_{hi}m_{hi}\boldsymbol{\Delta}(\boldsymbol{\pi}_{hi}^{\ast}\left(
\boldsymbol{\beta}\right)  )\otimes\boldsymbol{x}_{hi}\boldsymbol{x}_{hi}%
^{T},
\]%
\[
\mathbf{G}_{n}\left(  \boldsymbol{\beta}\right)  =\frac{1}{n}%
{\displaystyle\sum\limits_{h=1}^{H}}
{\displaystyle\sum\limits_{i=1}^{n_{h}}}
\boldsymbol{V}[\boldsymbol{U}_{hi}\left(  \boldsymbol{\beta}\right)
],\text{\quad}\mathbf{G}_{n_{h}}^{(h)}\left(  \boldsymbol{\beta}\right)
=\frac{1}{n_{h}}%
{\displaystyle\sum\limits_{i=1}^{n_{h}}}
\boldsymbol{V}[\boldsymbol{U}_{hi}\left(  \boldsymbol{\beta}\right)
],\text{\quad}\boldsymbol{V}[\boldsymbol{U}_{hi}\left(  \boldsymbol{\beta
}\right)  ]=w_{hi}^{2}\boldsymbol{V}[\widehat{\boldsymbol{Y}}_{hi}^{\ast
}]\otimes\boldsymbol{x}_{hi}\boldsymbol{x}_{hi}^{T},
\]
$\mathbf{H}\left(  \boldsymbol{\beta}\right)  $ is the Fisher information
matrix, $\boldsymbol{V}[\boldsymbol{\cdot}]$ denotes the variance-covariance
matrix of a random vector and $\boldsymbol{U}_{hi}\left(  \boldsymbol{\beta
}\right)  $ is the random variable generator of $\boldsymbol{u}_{hi}\left(
\boldsymbol{\beta}\right)  $, given by (\ref{Un2}).
\end{theorem}

\begin{proof}
From Theorem \ref{Th0} and by following the same steps of the linearization
method of Binder (1983),
\[
\mathbf{G}\left(  \boldsymbol{\beta}\right)  =\lim_{n\rightarrow\infty
}\boldsymbol{V}[\tfrac{1}{\sqrt{n}}\boldsymbol{U}_{\phi}\left(
\boldsymbol{\beta}\right)  ]\quad\text{and}\quad\mathbf{H}\left(
\boldsymbol{\beta}\right)  =-\lim_{n\rightarrow\infty}\frac{1}{n}%
\frac{\partial\boldsymbol{U}_{\phi}^{T}\left(  \boldsymbol{\beta}\right)
}{\partial\boldsymbol{\beta}},
\]
where $\boldsymbol{U}_{\phi}\left(  \boldsymbol{\beta}\right)  $ is the random
vector generator of $\boldsymbol{u}_{\phi}\left(  \boldsymbol{\beta}\right)
$, given by (\ref{3.09}). Taking into account that $f_{\phi,his}(\pi
_{his}(\boldsymbol{\beta}),\boldsymbol{\beta})=0$ and $f_{\phi,his}^{\prime
}(\pi_{his}(\boldsymbol{\beta}),\boldsymbol{\beta})=\frac{1}{\pi
_{his}(\boldsymbol{\beta})}\phi^{\prime\prime}\left(  1\right)  $, a first
Taylor expansion of $f_{\phi,his}(\tfrac{\widehat{Y}_{his}}{m_{hi}%
},\boldsymbol{\beta})$ given in (\ref{3.010b}) is%
\begin{align}
f_{\phi,his}(\tfrac{\widehat{Y}_{his}}{m_{hi}},\boldsymbol{\beta})  &
=f_{\phi,his}(\pi_{his}(\boldsymbol{\beta}),\boldsymbol{\beta})+f_{\phi
,his}^{\prime}(\pi_{his}(\boldsymbol{\beta}),\boldsymbol{\beta})(\tfrac
{\widehat{Y}_{his}}{m_{hi}}-\pi_{his}(\boldsymbol{\beta}))+o(\tfrac
{\widehat{Y}_{his}}{m_{hi}}-\pi_{his}(\boldsymbol{\beta}))\nonumber\\
&  =\frac{\phi^{\prime\prime}\left(  1\right)  }{\pi_{his}(\boldsymbol{\beta
})}(\tfrac{Y_{his}}{m_{hi}}-\pi_{his}(\boldsymbol{\beta}))+o(\tfrac
{\widehat{Y}_{his}}{m_{hi}}-\pi_{his}(\boldsymbol{\beta})), \label{ff}%
\end{align}
i.e.%
\[
\boldsymbol{f}_{\phi,hi}(\tfrac{\widehat{\boldsymbol{Y}}_{hi}}{m_{hi}%
},\boldsymbol{\beta})=\phi^{\prime\prime}\left(  1\right)  \mathrm{diag}%
^{-1}(\boldsymbol{\pi}_{hi}(\boldsymbol{\beta}))(\tfrac
{\widehat{\boldsymbol{Y}}_{hi}}{m_{hi}}-\boldsymbol{\pi}_{hi}%
(\boldsymbol{\beta}))+o\left(  \boldsymbol{1}_{d+1}\left\Vert \tfrac
{\widehat{\boldsymbol{Y}}_{hi}}{m_{hi}}-\boldsymbol{\pi}_{hi}%
(\boldsymbol{\beta})\right\Vert \right)  ,
\]
and hence from (\ref{uu})%
\[
\frac{1}{\sqrt{n}}\boldsymbol{U}_{\phi}\left(  \boldsymbol{\beta}\right)
=\frac{1}{\sqrt{n}}\sum\limits_{h=1}^{H}\sum\limits_{i=1}^{n_{h}}w_{hi}%
m_{hi}\frac{\partial\boldsymbol{\pi}_{hi}^{T}(\boldsymbol{\beta})}%
{\partial\boldsymbol{\beta}}\mathrm{diag}^{-1}(\boldsymbol{\pi}_{hi}%
(\boldsymbol{\beta}))(\tfrac{\widehat{\boldsymbol{Y}}_{hi}}{m_{hi}%
}-\boldsymbol{\pi}_{hi}(\boldsymbol{\beta}))+%
{\displaystyle\sum\limits_{h=1}^{H}}
\sqrt{\eta_{h}^{\ast}}o\left(  \boldsymbol{1}_{dk}\left\Vert \frac{1}%
{\sqrt{n_{h}}}\left(
{\displaystyle\sum_{i=1}^{n_{h}}}
\widehat{\boldsymbol{Y}}_{hi}-%
{\displaystyle\sum_{i=1}^{n_{h}}}
m_{hi}\boldsymbol{\pi}_{hi}(\boldsymbol{\beta})\right)  \right\Vert \right)
.
\]
From the Central Limit Theorem given in Rao (1973, page 147)
\[
\frac{1}{\sqrt{n_{h}}}\left(
{\displaystyle\sum_{i=1}^{n_{h}}}
\widehat{\boldsymbol{Y}}_{hi}-%
{\displaystyle\sum_{i=1}^{n_{h}}}
m_{hi}\boldsymbol{\pi}_{hi}(\boldsymbol{\beta})\right)  \underset{n_{h}%
\rightarrow\infty}{\overset{\mathcal{L}}{\longrightarrow}}\mathcal{N}%
(\boldsymbol{0}_{d+1},\lim_{n_{h}\rightarrow\infty}\tfrac{1}{n_{h}}%
{\textstyle\sum_{i=1}^{n_{h}}}
\boldsymbol{V}[\widehat{\boldsymbol{Y}}_{hi}]),
\]
then
\[
o\left(  \boldsymbol{1}_{dk}\left\Vert \frac{1}{\sqrt{n_{h}}}\left(
{\displaystyle\sum_{i=1}^{n_{h}}}
\widehat{\boldsymbol{Y}}_{hi}-%
{\displaystyle\sum_{i=1}^{n_{h}}}
m_{hi}\boldsymbol{\pi}_{hi}(\boldsymbol{\beta})\right)  \right\Vert \right)
=o\left(  o_{p}(\boldsymbol{1}_{dk})\right)  =o_{p}(\boldsymbol{1}_{dk}),
\]
and thus%
\[
\frac{1}{\sqrt{n}}\boldsymbol{U}_{\phi}\left(  \boldsymbol{\beta}\right)
=\frac{1}{\sqrt{n}}\sum\limits_{h=1}^{H}\sum\limits_{i=1}^{n_{h}}w_{hi}%
\frac{\partial\log\boldsymbol{\pi}_{hi}^{T}(\boldsymbol{\beta})}%
{\partial\boldsymbol{\beta}}(\widehat{\boldsymbol{y}}_{hi}-m_{hi}%
\boldsymbol{\pi}_{hi}(\boldsymbol{\beta}))+o_{p}(\boldsymbol{1}_{dk}).
\]
Since%
\begin{align*}
\frac{\partial\log\boldsymbol{\pi}_{hi}^{T}(\boldsymbol{\beta})}%
{\partial\boldsymbol{\beta}}\boldsymbol{\pi}_{hi}(\boldsymbol{\beta})  &
=\frac{\partial\boldsymbol{\pi}_{hi}^{T}(\boldsymbol{\beta})}{\partial
\boldsymbol{\beta}}\mathrm{diag}^{-1}(\boldsymbol{\pi}_{hi}(\boldsymbol{\beta
}))\boldsymbol{\pi}_{hi}(\boldsymbol{\beta})\\
&  =\frac{\partial\boldsymbol{\pi}_{hi}^{T}(\boldsymbol{\beta})}%
{\partial\boldsymbol{\beta}}\boldsymbol{1}_{d+1}=\frac{\partial\left(
\boldsymbol{\pi}_{hi}^{T}(\boldsymbol{\beta})\boldsymbol{1}_{d+1}\right)
}{\partial\boldsymbol{\beta}}=\boldsymbol{0}_{dk},
\end{align*}%
\begin{align*}
\frac{\partial\log\boldsymbol{\pi}_{hi}^{T}(\boldsymbol{\beta})}%
{\partial\boldsymbol{\beta}}\widehat{\boldsymbol{y}}_{hi}  &  =\frac
{\partial\boldsymbol{\pi}_{hi}^{T}(\boldsymbol{\beta})}{\partial
\boldsymbol{\beta}}\mathrm{diag}^{-1}(\boldsymbol{\pi}_{hi}(\boldsymbol{\beta
}))\widehat{\boldsymbol{Y}}_{hi}\\
&  =\left(  \left(  \boldsymbol{I}_{d\times d},\boldsymbol{0}_{d\times
1}\right)  \boldsymbol{\Delta}(\boldsymbol{\pi}_{hi}\left(  \boldsymbol{\beta
}\right)  )\otimes\boldsymbol{x}_{hi}\right)  \mathrm{diag}^{-1}%
(\boldsymbol{\pi}_{hi}(\boldsymbol{\beta}))\widehat{\boldsymbol{Y}}_{hi}\\
&  =\left(  \boldsymbol{I}_{d\times d},\boldsymbol{0}_{d\times1}\right)
\boldsymbol{\Delta}(\boldsymbol{\pi}_{hi}\left(  \boldsymbol{\beta}\right)
)\mathrm{diag}^{-1}(\boldsymbol{\pi}_{hi}(\boldsymbol{\beta}%
))\widehat{\boldsymbol{Y}}_{hi}\otimes\boldsymbol{x}_{hi}\\
&  =\left(  \boldsymbol{I}_{d\times d},\boldsymbol{0}_{d\times1}\right)
\left(  \widehat{\boldsymbol{Y}}_{hi}-\boldsymbol{\pi}_{hi}\left(
\boldsymbol{\beta}\right)  \boldsymbol{\pi}_{hi}\left(  \boldsymbol{\beta
}\right)  ^{T}\mathrm{diag}^{-1}\left(  \pi_{hi}\left(  \boldsymbol{\beta
}\right)  \right)  \widehat{\boldsymbol{Y}}_{hi}\right)  \otimes
\boldsymbol{x}_{hi}\\
&  =\left(  \boldsymbol{I}_{d\times d},\boldsymbol{0}_{d\times1}\right)
\left(  \widehat{\boldsymbol{Y}}_{hi}-m_{hi}\boldsymbol{\pi}_{hi}\left(
\boldsymbol{\beta}\right)  \right)  \otimes\boldsymbol{x}_{hi}\\
&  =\left(  \widehat{\boldsymbol{y}}_{hi}^{\ast}-m_{hi}\boldsymbol{\pi}%
_{hi}^{\ast}\left(  \boldsymbol{\beta}\right)  \right)  \otimes\boldsymbol{x}%
_{hi},
\end{align*}
it follows that%
\begin{equation}
\frac{1}{\sqrt{n}}\boldsymbol{U}_{\phi}\left(  \boldsymbol{\beta}\right)
=\frac{1}{\sqrt{n}}\sum\limits_{h=1}^{H}\sum\limits_{i=1}^{n_{h}}w_{hi}\left(
\widehat{\boldsymbol{y}}_{hi}^{\ast}-m_{hi}\boldsymbol{\pi}_{hi}^{\ast
}(\boldsymbol{\beta})\right)  \otimes\boldsymbol{x}_{hi}+o_{p}(\boldsymbol{1}%
_{dk}), \label{Texp}%
\end{equation}
Then $\mathbf{H}\left(  \boldsymbol{\beta}_{0}\right)  $ is the limit of%
\begin{align*}
-\frac{1}{n}\frac{\partial}{\partial\boldsymbol{\beta}}\boldsymbol{U}_{\phi
}^{T}\left(  \boldsymbol{\beta}\right)   &  =\frac{1}{n}\sum\limits_{h=1}%
^{H}\sum\limits_{i=1}^{n_{h}}w_{hi}m_{hi}\frac{\partial}{\partial
\boldsymbol{\beta}}\boldsymbol{\pi}_{hi}^{\ast}(\boldsymbol{\beta}%
)\otimes\boldsymbol{x}_{hi}+o_{p}(\boldsymbol{1}_{dk\times dk})\\
&  =\frac{1}{n}\sum\limits_{h=1}^{H}\sum\limits_{i=1}^{n_{h}}w_{hi}%
m_{hi}\boldsymbol{\Delta}(\boldsymbol{\pi}_{hi}^{\ast}\left(
\boldsymbol{\beta}\right)  )\otimes\boldsymbol{x}_{hi}+o_{p}(\boldsymbol{1}%
_{dk\times dk}),
\end{align*}
as $n$ increases, and hence $\mathbf{H}\left(  \boldsymbol{\beta}\right)
=\lim_{n\rightarrow\infty}\mathbf{H}_{n}\left(  \boldsymbol{\beta}\right)  $.
On the other hand, from (\ref{Texp}) it follows that%
\[
\frac{1}{\sqrt{n}}\boldsymbol{U}_{\phi}\left(  \boldsymbol{\beta}\right)
=\frac{1}{\sqrt{n}}\boldsymbol{U}\left(  \boldsymbol{\beta}\right)
+o_{p}(\boldsymbol{1}_{dk}),
\]
and this justifies that $\mathbf{G}\left(  \boldsymbol{\beta}\right)
=\lim_{n\rightarrow\infty}\mathbf{G}_{n}\left(  \boldsymbol{\beta}\right)  $.
\end{proof}

The following result justifies how to estimate $\mathbf{G}_{n}\left(
\boldsymbol{\beta}\right)  $, in particular $\widehat{\mathbf{G}}%
_{n}(\widehat{\boldsymbol{\beta}}_{P})$ given in (\ref{GHat0}), which is
provided by the \texttt{SURVEYLOGISTIC} procedure\ of \texttt{SAS}.

\begin{remark}
\label{Th2}Matrix $\mathbf{G}\left(  \boldsymbol{\beta}_{0}\right)  $ of
Theorem \ref{Th1} can be consistently estimated as%
\begin{equation}
\widehat{\mathbf{G}}_{n}(\widehat{\boldsymbol{\beta}}_{\phi,P})=\frac{1}{n}%
{\displaystyle\sum\limits_{h=1}^{H}}
{\displaystyle\sum\limits_{i=1}^{n_{h}}}
\left(  \boldsymbol{u}_{hi}(\widehat{\boldsymbol{\beta}}_{\phi,P})-\tfrac
{1}{n}\boldsymbol{u}(\widehat{\boldsymbol{\beta}}_{\phi,P})\right)  \left(
\boldsymbol{u}_{hi}(\widehat{\boldsymbol{\beta}}_{\phi,P})-\tfrac{1}%
{n}\boldsymbol{u}(\widehat{\boldsymbol{\beta}}_{\phi,P})\right)  ^{T}
\label{GHat}%
\end{equation}
with $\widehat{\boldsymbol{\beta}}_{\phi,P}$ being any pseudo minimum
phi-divergence estimator of parameter $\boldsymbol{\beta}$. In particular, if
$\phi(x)=x\log x-x+1$,%
\begin{equation}
\widehat{\mathbf{G}}_{n}(\widehat{\boldsymbol{\beta}}_{P})=\frac{1}{n}%
{\displaystyle\sum\limits_{h=1}^{H}}
{\displaystyle\sum\limits_{i=1}^{n_{h}}}
\boldsymbol{u}_{hi}(\widehat{\boldsymbol{\beta}}_{P})\boldsymbol{u}_{hi}%
^{T}(\widehat{\boldsymbol{\beta}}_{P}), \label{GHat0}%
\end{equation}
since $\boldsymbol{u}(\widehat{\boldsymbol{\beta}}_{P})=\boldsymbol{0}_{dk}$.
On the other hand, matrix $\mathbf{H}\left(  \boldsymbol{\beta}_{0}\right)  $
of Theorem \ref{Th1} can be consistently estimated as%
\[
\mathbf{H}_{n}(\widehat{\boldsymbol{\beta}}_{\phi,P})=\frac{1}{n}%
{\displaystyle\sum\limits_{h=1}^{H}}
{\displaystyle\sum\limits_{i=1}^{n_{h}}}
w_{hi}m_{hi}\boldsymbol{\Delta}(\boldsymbol{\pi}_{hi}^{\ast}%
(\widehat{\boldsymbol{\beta}}_{\phi,P}))\otimes\boldsymbol{x}_{hi}%
\boldsymbol{x}_{hi}^{T}.
\]

\end{remark}

Let $\widehat{\boldsymbol{\beta}}_{\phi}$ denote the minimum phi-divergence
estimator of $\boldsymbol{\beta}$ for simple random sampling within each
cluster, i.e. multinomial sampling. By following Gupta and Pardo (2007), it
can be seen that
\[
\lim_{n\rightarrow\infty}\boldsymbol{V}[\sqrt{n}\widehat{\boldsymbol{\beta}%
}_{\phi}]=\mathbf{H}^{-1}\left(  \boldsymbol{\beta}_{0}\right)  .
\]
The \textquotedblleft design effect matrix\textquotedblright\ for the
multinomial logistic regression model with sample survey design is defined as
$\lim_{n\rightarrow\infty}\boldsymbol{V}[\sqrt{n}\widehat{\boldsymbol{\beta}%
}_{\phi,P}]\boldsymbol{V}^{-1}[\sqrt{n}\widehat{\boldsymbol{\beta}}_{\phi
}]=\mathbf{H}^{-1}\left(  \boldsymbol{\beta}_{0}\right)  \mathbf{G}\left(
\boldsymbol{\beta}_{0}\right)  $ and the \textquotedblleft design
effect\textquotedblright, denoted by $\nu$, for the multinomial logistic
regression model with sample survey design is defined as $\nu\mathbf{(}%
\boldsymbol{\beta}_{0})=\frac{1}{dk}\mathrm{trace}\left(  \mathbf{H}%
^{-1}\mathbf{(}\boldsymbol{\beta}_{0})\mathbf{G(}\boldsymbol{\beta}%
_{0})\right)  $. In practice, $\mathbf{H(}\boldsymbol{\beta}_{0})$ and
$\mathbf{G(}\boldsymbol{\beta}_{0})$\ can be consistently estimated through
the pseudo minimum phi-divergence estimator of parameter $\boldsymbol{\beta}$
as%
\[
\mathbf{H}_{n}(\widehat{\boldsymbol{\beta}}_{\phi,P})=\frac{1}{n}%
{\displaystyle\sum\limits_{h=1}^{H}}
{\displaystyle\sum\limits_{i=1}^{n_{h}}}
w_{hi}m_{hi}\boldsymbol{\Delta}(\boldsymbol{\pi}_{hi}^{\ast}%
(\widehat{\boldsymbol{\beta}}_{\phi,P}))\otimes\boldsymbol{x}_{hi}%
\boldsymbol{x}_{hi}^{T},
\]
and $\widehat{\mathbf{G}}_{n}(\widehat{\boldsymbol{\beta}}_{\phi,P})$\ given
in (\ref{GHat}). For more details about the design matrix in other models see
for instance Rao and Scott (1984) or formula 7.6 in Rao and Thomas (1989).

\begin{definition}
A consistent estimator of the design effect matrix, $\mathbf{H}^{-1}\left(
\boldsymbol{\beta}\right)  \mathbf{G}\left(  \boldsymbol{\beta}\right)  $,
based on the linearization method of Binder (1983) and the pseudo minimum
phi-divergence estimator of parameter $\boldsymbol{\beta}$, is%
\begin{align*}
\mathbf{H}_{n}^{-1}(\widehat{\boldsymbol{\beta}}_{\phi,P})\widehat{\mathbf{G}%
}_{n}(\widehat{\boldsymbol{\beta}}_{\phi,P})  &  =\left(
{\displaystyle\sum\limits_{h=1}^{H}}
{\displaystyle\sum\limits_{i=1}^{n_{h}}}
w_{hi}m_{hi}\boldsymbol{\Delta}(\boldsymbol{\pi}_{hi}^{\ast}%
(\widehat{\boldsymbol{\beta}}_{\phi,P}))\otimes\boldsymbol{x}_{hi}%
\boldsymbol{x}_{hi}^{T}\right)  ^{-1}\\
&  \times%
{\displaystyle\sum\limits_{h=1}^{H}}
{\displaystyle\sum\limits_{i=1}^{n_{h}}}
\left(  \boldsymbol{u}_{hi}(\widehat{\boldsymbol{\beta}}_{\phi,P})-\tfrac
{1}{n}\boldsymbol{u}(\widehat{\boldsymbol{\beta}}_{\phi,P})\right)  \left(
\boldsymbol{u}_{hi}(\widehat{\boldsymbol{\beta}}_{\phi,P})-\tfrac{1}%
{n}\boldsymbol{u}(\widehat{\boldsymbol{\beta}}_{\phi,P})\right)  ^{T}.
\end{align*}
Similarly, a consistent estimator of the design effect, $\nu\left(
\boldsymbol{\beta}_{0}\right)  =\frac{1}{dk}\mathrm{trace}\left(
\mathbf{H}^{-1}\left(  \boldsymbol{\beta}_{0}\right)  \mathbf{G}\left(
\boldsymbol{\beta}_{0}\right)  \right)  $, based on the linearization method
of Binder (1983) and the pseudo minimum phi-divergence estimator of parameter
$\boldsymbol{\beta}$, is%
\begin{equation}
\widehat{\nu}(\widehat{\boldsymbol{\beta}}_{\phi,P})=\frac{1}{dk}%
\mathrm{trace}\left(  \mathbf{H}_{n}^{-1}(\widehat{\boldsymbol{\beta}}%
_{\phi,P})\widehat{\mathbf{G}}_{n}(\widehat{\boldsymbol{\beta}}_{\phi
,P})\right)  . \label{eff}%
\end{equation}

\end{definition}

The estimator of the design effect is specially interesting for clusters such
that
\begin{align}
\boldsymbol{E}[\widehat{\boldsymbol{Y}}_{hi}]  &  =m_{h}\boldsymbol{\pi}%
_{hi}\left(  \boldsymbol{\beta}_{0}\right)  \quad\text{and}\quad
\boldsymbol{V}[\widehat{\boldsymbol{Y}}_{hi}]=\nu_{m_{h}}m_{h}%
\boldsymbol{\Delta}(\boldsymbol{\pi}_{hi}\left(  \boldsymbol{\beta}%
_{0}\right)  ),\label{multOver}\\
\nu_{m_{h}}  &  =1+\rho_{h}^{2}(m_{h}-1),\nonumber
\end{align}
with $\nu_{m_{h}}$ being the overdispersion parameter,$\ \rho_{h}^{2}$\ being
the intra-cluster correlation coefficient and equal cluster sizes in the
strata, $m_{hi}=m_{h}$, $h=1,...,H$, $i=1,...,n_{h}$. Examples of
distributions of $\widehat{\boldsymbol{y}}_{hi}$ verifying (\ref{multOver})
are the so-called \textquotedblleft overdispersed multinomial
distributions\textquotedblright\ (see Alonso et al. (2016)). For these
distributions, once the pseudo minimum phi-divergence estimator of parameter
$\boldsymbol{\beta}$, $\widehat{\boldsymbol{\beta}}_{\phi,P}$, is obtained,
the interest lies in estimating $\rho_{h}^{2}$. In Theorems \ref{Th3} and
\ref{Th4} two proposals of families of estimates for $\nu_{m_{h}}$ and
$\rho_{h}^{2}$ are established. Both proposals are independent of the weights
except for $\widehat{\boldsymbol{\beta}}_{\phi,P}$, and this fact has a
logical explanation taking into account that the weights are constructed only
for estimation of $\boldsymbol{\beta}$.

\begin{theorem}
\label{Th3}Let $\widehat{\boldsymbol{\beta}}_{\phi,P}$ the pseudo minimum
phi-divergence estimate of parameter $\boldsymbol{\beta}$ for a multinomial
logistic regression model with \textquotedblleft overdispersed multinomial
distribution\textquotedblright. Assume that $w_{hi}=w_{h}$, $i=1,...,n_{h}$.
Then%
\begin{align}
\widehat{\nu}_{m_{h}}(\widehat{\boldsymbol{\beta}}_{\phi,P})  &  =\frac{1}%
{dk}\mathrm{trace}\left(  \left(
{\displaystyle\sum\limits_{i=1}^{n_{h}}}
m_{h}\boldsymbol{\Delta}(\boldsymbol{\pi}_{hi}^{\ast}%
(\widehat{\boldsymbol{\beta}}_{\phi,P}))\otimes\boldsymbol{x}_{hi}%
\boldsymbol{x}_{hi}^{T}\right)  ^{-1}\right. \nonumber\\
&  \left.  \times%
{\displaystyle\sum\limits_{i=1}^{n_{h}}}
\left(  \boldsymbol{v}_{hi}(\widehat{\boldsymbol{\beta}}_{\phi,P}%
)-\boldsymbol{\bar{v}}_{h}(\widehat{\boldsymbol{\beta}}_{\phi,P})\right)
\left(  \boldsymbol{v}_{hi}(\widehat{\boldsymbol{\beta}}_{\phi,P}%
)-\boldsymbol{\bar{v}}_{h}(\widehat{\boldsymbol{\beta}}_{\phi,P})\right)
^{T}\right)  \label{overdisp}%
\end{align}
with%
\begin{align*}
\boldsymbol{v}_{hi}(\widehat{\boldsymbol{\beta}}_{\phi,P})  &  =\boldsymbol{r}%
_{hi}^{\ast}\left(  \boldsymbol{\beta}\right)  \otimes\boldsymbol{x}_{hi},\\
\boldsymbol{\bar{v}}_{h}(\widehat{\boldsymbol{\beta}}_{\phi,P})  &  =\tfrac
{1}{n_{h}}%
{\displaystyle\sum\limits_{k=1}^{n_{h}}}
\boldsymbol{v}_{hk}(\widehat{\boldsymbol{\beta}}_{\phi,P}),
\end{align*}
is an estimator of $\nu_{m_{h}}$ based on the \textquotedblleft linearization
method of Binder\textquotedblright\ and the pseudo minimum phi-divergence
estimator of $\widehat{\boldsymbol{\beta}}_{\phi,P}$, and%
\[
\widehat{\rho}_{h}^{2}(\widehat{\boldsymbol{\beta}}_{\phi,P})=\frac
{\widehat{\nu}_{m_{h}}(\widehat{\boldsymbol{\beta}}_{\phi,P})-1}{m_{h}-1}%
\]
is an estimator of $\rho_{h}^{2}$ based on the \textquotedblleft linearization
method of Binder\textquotedblright\ and the pseudo minimum phi-divergence
estimator of $\widehat{\boldsymbol{\beta}}_{\phi,P}$.
\end{theorem}

\begin{proof}
If $\boldsymbol{V}[\widehat{\boldsymbol{Y}}_{hi}]=\nu_{m_{h}}m_{h}%
\boldsymbol{\Delta}(\boldsymbol{\pi}_{hi}\left(  \boldsymbol{\beta}%
_{0}\right)  )$, then from the expression of $\mathbf{G}_{n_{h}}^{(h)}\left(
\boldsymbol{\beta}_{0}\right)  $ given in Theorem \ref{Th2},
\begin{align*}
\mathbf{G}_{n_{h}}^{(h)}\left(  \boldsymbol{\beta}_{0}\right)   &  =\frac
{1}{n_{h}}%
{\displaystyle\sum\limits_{i=1}^{n_{h}}}
w_{h}^{2}\boldsymbol{V}[\widehat{\boldsymbol{Y}}_{hi}^{\ast}]\otimes
\boldsymbol{x}_{hi}\boldsymbol{x}_{hi}^{T}=\nu_{m_{h}}w_{h}\frac{1}{n_{h}}%
{\displaystyle\sum\limits_{i=1}^{n_{h}}}
w_{h}m_{h}\boldsymbol{\Delta}(\boldsymbol{\pi}_{hi}^{\ast}\left(
\boldsymbol{\beta}_{0}\right)  )\otimes\boldsymbol{x}_{hi}\boldsymbol{x}%
_{hi}^{T}\\
&  =\nu_{m_{h}}w_{h}\mathbf{H}_{n_{h}}^{(h)}\left(  \boldsymbol{\beta}%
_{0}\right)  .
\end{align*}
Hence, from%
\[
\mathrm{trace}\left(  \mathbf{H}_{n_{h}}^{(h)}\left(  \boldsymbol{\beta}%
_{0}\right)  ^{-1}\mathbf{G}_{n_{h}}^{(h)}\left(  \boldsymbol{\beta}%
_{0}\right)  \right)  =\nu_{m_{h}}w_{h}dk,
\]
and consistency of $\mathbf{H}_{n_{h}}^{(h)}(\widehat{\boldsymbol{\beta}%
}_{\phi,P})$ and $\widehat{\mathbf{G}}_{n_{h}}^{(h)}%
(\widehat{\boldsymbol{\beta}}_{\phi,P})$,
\[
\widehat{\nu}_{m_{h}}(\widehat{\boldsymbol{\beta}}_{\phi,P})=\frac{1}%
{dk}\mathrm{trace}\left(  \frac{1}{w_{h}}\mathbf{H}_{n_{h}}^{(h)}%
(\widehat{\boldsymbol{\beta}}_{\phi,P})^{-1}\widehat{\mathbf{G}}_{n_{h}}%
^{(h)}(\widehat{\boldsymbol{\beta}}_{\phi,P})\right)  ,
\]
is proven with%
\begin{align*}
\frac{1}{w_{h}}\mathbf{H}_{n_{h}}^{(h)}(\widehat{\boldsymbol{\beta}}_{\phi
,P})^{-1}\widehat{\mathbf{G}}_{n_{h}}^{(h)}(\widehat{\boldsymbol{\beta}}%
_{\phi,P})  &  =\left(
{\displaystyle\sum\limits_{i=1}^{n_{h}}}
m_{h}\boldsymbol{\Delta}(\boldsymbol{\pi}_{hi}^{\ast}%
(\widehat{\boldsymbol{\beta}}_{\phi,P}))\otimes\boldsymbol{x}_{hi}%
\boldsymbol{x}_{hi}^{T}\right)  ^{-1}\\
&  \times%
{\displaystyle\sum\limits_{i=1}^{n_{h}}}
\left(  \boldsymbol{v}_{hi}(\widehat{\boldsymbol{\beta}}_{\phi,P}%
)-\boldsymbol{\bar{v}}_{h}(\widehat{\boldsymbol{\beta}}_{\phi,P})\right)
\left(  \boldsymbol{v}_{hi}(\widehat{\boldsymbol{\beta}}_{\phi,P}%
)-\boldsymbol{\bar{v}}_{h}(\widehat{\boldsymbol{\beta}}_{\phi,P})\right)
^{T},\\
\boldsymbol{v}_{hi}(\widehat{\boldsymbol{\beta}}_{\phi,P})  &  =\frac{1}%
{w_{h}}\boldsymbol{u}_{hi}\left(  \boldsymbol{\beta}\right)  ,
\end{align*}
which is equivalent to (\ref{overdisp}).
\end{proof}

\begin{remark}
Since
\begin{equation}
\widehat{\nu}_{m_{h}}(\widehat{\boldsymbol{\beta}}_{\phi,P})=\frac{1}{w_{h}%
}\frac{1}{dk}\mathrm{trace}\left(  \mathbf{H}_{n_{h}}^{(h)}%
(\widehat{\boldsymbol{\beta}}_{\phi,P})^{-1}\widehat{\mathbf{G}}_{n_{h}}%
^{(h)}(\widehat{\boldsymbol{\beta}}_{\phi,P})\right)  =\frac{1}{w_{h}%
}\widehat{\nu}^{(h)}(\widehat{\boldsymbol{\beta}}_{\phi,P}), \label{overdisp2}%
\end{equation}
unless $w_{h}=1$, the overdispersion parameter $\widehat{\nu}_{m_{h}%
}(\widehat{\boldsymbol{\beta}}_{\phi,P})$ and the design effect $\widehat{\nu
}^{(h)}(\widehat{\boldsymbol{\beta}}_{\phi,P})$ of the $h$-th stratum are not
in general equivalent. Based on the expression of (\ref{overdisp})
$\widehat{\nu}_{m_{h}}(\cdot)$, does not depend on the weights except for that
$\widehat{\boldsymbol{\beta}}_{\phi,P}$ is plugged in $\widehat{\nu}_{m_{h}%
}(\cdot)$, additionally based on (\ref{overdisp2}) it is concluded that
$\widehat{\nu}^{(h)}(\widehat{\boldsymbol{\beta}}_{\phi,P})$\ is directly
proportional to the weights.
\end{remark}

\begin{theorem}
\label{Th4}Let $\widehat{\boldsymbol{\beta}}_{\phi,P}$ the pseudo minimum
phi-divergence estimate of parameter $\boldsymbol{\beta}$ for a multinomial
logistic regression model with \textquotedblleft overdispersed multinomial
distribution\textquotedblright. Then%
\[
\widetilde{\nu}_{m_{h}}(\widehat{\boldsymbol{\beta}}_{\phi,P})=\frac{1}%
{n_{h}d}\sum\limits_{i=1}^{n_{h}}\sum\limits_{s=1}^{d+1}\frac{\left(
\widehat{y}_{his}-m_{h}\pi_{his}(\widehat{\boldsymbol{\beta}}_{\phi
,P})\right)  ^{2}}{m_{h}\pi_{his}(\widehat{\boldsymbol{\beta}}_{\phi,P})}%
\]
is an estimation of $\nu_{m_{h}}$ based on the \textquotedblleft method of
moments\textquotedblright\ and the pseudo minimum phi-divergence estimator of
$\widehat{\boldsymbol{\beta}}_{\phi,P}$, and%
\[
\widetilde{\rho}_{h}^{2}(\widehat{\boldsymbol{\beta}}_{\phi,P})=\frac
{\widetilde{\nu}_{m_{h}}(\widehat{\boldsymbol{\beta}}_{\phi,P})-1}{m_{h}-1}%
\]
is an estimation of $\rho_{h}^{2}$ based on the \textquotedblleft method of
moments\textquotedblright\ and the pseudo minimum phi-divergence estimator of
$\widehat{\boldsymbol{\beta}}_{\phi,P}$.
\end{theorem}

\begin{proof}
The mean vector and variance-covariance matrix of
\[
\boldsymbol{Z}_{hi}^{\ast}(\boldsymbol{\beta}_{0})=\sqrt{m_{h}}%
\boldsymbol{\Delta}^{-\frac{1}{2}}(\boldsymbol{\pi}_{hi}^{\ast}\left(
\boldsymbol{\beta}_{0}\right)  )(\tfrac{\widehat{\boldsymbol{Y}}_{hi}^{\ast}%
}{m_{h}}-\boldsymbol{\pi}_{hi}^{\ast}\left(  \boldsymbol{\beta}_{0}\right)
),
\]
are respectively%
\begin{align*}
\boldsymbol{E}[\boldsymbol{Z}_{hi}^{\ast}(\boldsymbol{\beta}_{0})]  &
=\boldsymbol{0}_{d},\\
\boldsymbol{V}[\boldsymbol{Z}_{hi}^{\ast}(\boldsymbol{\beta}_{0})]  &
=\nu_{m_{h}}\boldsymbol{I}_{d},
\end{align*}
for $h=1,...,H$. An unbiased estimator of $\boldsymbol{V}[\boldsymbol{Z}%
_{hi}^{\ast}(\boldsymbol{\beta}_{0})]$ is%
\[
\widehat{\boldsymbol{V}}[\boldsymbol{Z}_{hi}^{\ast}(\boldsymbol{\beta}%
_{0})]=\frac{1}{n_{h}}\sum\limits_{i=1}^{n_{h}}\boldsymbol{Z}_{hi}^{\ast
}(\boldsymbol{\beta}_{0})\boldsymbol{Z}_{hi}^{\ast T}(\boldsymbol{\beta}%
_{0}),
\]
from which is derived%
\begin{align*}
E\left[  \mathrm{trace}\widehat{\boldsymbol{V}}[\boldsymbol{Z}_{hi}^{\ast
}(\boldsymbol{\beta}_{0})]\right]   &  =\mathrm{trace}\boldsymbol{V}%
[\boldsymbol{Z}_{hi}^{\ast}(\boldsymbol{\beta}_{0})],\\
E\left[  \frac{1}{n_{h}}\sum\limits_{i=1}^{n_{h}}\mathrm{trace}\left(
\boldsymbol{Z}_{hi}^{\ast}(\boldsymbol{\beta}_{0})\boldsymbol{Z}_{hi}^{\ast
T}(\boldsymbol{\beta}_{0})\right)  \right]   &  =\mathrm{trace}\left(
\nu_{m_{h}}\boldsymbol{I}_{d}\right)  ,\\
E\left[  \frac{1}{n_{h}}\sum\limits_{i=1}^{n_{h}}\boldsymbol{Z}_{hi}^{\ast
T}(\boldsymbol{\beta}_{0})\boldsymbol{Z}_{hi}^{\ast}(\boldsymbol{\beta}%
_{0})\right]   &  =\nu_{m_{h}}d,\\
E\left[  \frac{1}{n_{h}d}\sum\limits_{i=1}^{n_{h}}\boldsymbol{Z}_{hi}^{\ast
T}(\boldsymbol{\beta}_{0})\boldsymbol{Z}_{hi}^{\ast}(\boldsymbol{\beta}%
_{0})\right]   &  =\nu_{m_{h}}.
\end{align*}
This expression suggest using%
\begin{align*}
\widetilde{\nu}_{m_{h}}(\widehat{\boldsymbol{\beta}}_{\phi,P})  &  =\frac
{1}{n_{h}d}\sum\limits_{i=1}^{n_{h}}\widehat{\boldsymbol{z}}_{hi,\phi,P}^{\ast
T}(\widehat{\boldsymbol{\beta}}_{\phi,P})\widehat{\boldsymbol{z}}_{hi,\phi
,P}^{\ast}(\widehat{\boldsymbol{\beta}}_{\phi,P})\\
&  =\frac{1}{n_{h}d}m_{h}\left(  \tfrac{\widehat{\boldsymbol{y}}_{hi}^{\ast}%
}{m_{h}}-\boldsymbol{\pi}_{hi}^{\ast}(\widehat{\boldsymbol{\beta}}_{\phi
,P})\right)  ^{T}\boldsymbol{\Delta}^{-1}(\boldsymbol{\pi}_{hi}^{\ast
}(\widehat{\boldsymbol{\beta}}_{\phi,P}))\left(  \tfrac
{\widehat{\boldsymbol{y}}_{hi}^{\ast}}{m_{h}}-\boldsymbol{\pi}_{hi}^{\ast
}(\widehat{\boldsymbol{\beta}}_{\phi,P})\right) \\
&  =\frac{1}{n_{h}d}m_{h}\left(  \tfrac{\widehat{\boldsymbol{y}}_{hi}}{m_{h}%
}-\boldsymbol{\pi}_{hi}(\widehat{\boldsymbol{\beta}}_{\phi,P})\right)
^{T}\boldsymbol{\Delta}^{-}(\boldsymbol{\pi}_{hi}(\widehat{\boldsymbol{\beta}%
}_{\phi,P}))\left(  \tfrac{\widehat{\boldsymbol{y}}_{hi}^{\ast}}{m_{h}%
}-\boldsymbol{\pi}_{hi}^{\ast}(\widehat{\boldsymbol{\beta}}_{\phi,P})\right)
,\\
\widehat{\boldsymbol{z}}_{hi,\phi,P}^{\ast}  &  =\sqrt{m_{h}}%
\boldsymbol{\Delta}^{-\frac{1}{2}}(\boldsymbol{\pi}_{hi}^{\ast}%
(\widehat{\boldsymbol{\beta}}_{\phi,P}))\left(  \tfrac{\widehat{\boldsymbol{y}%
}_{hi}^{\ast}}{m_{h}}-\boldsymbol{\pi}_{hi}^{\ast}(\widehat{\boldsymbol{\beta
}}_{\phi,P})\right)  .
\end{align*}
Finally, since $\boldsymbol{\Delta}^{-}(\boldsymbol{\pi}_{hi}%
(\widehat{\boldsymbol{\beta}}_{\phi,P}))=\mathrm{diag}^{-1}(\boldsymbol{\pi
}_{hi}(\widehat{\boldsymbol{\beta}}_{\phi,P}))$, is a possible expression for
the generalized inverse, the desired result for $\widetilde{\nu}_{m_{h}%
}(\widehat{\boldsymbol{\beta}}_{\phi,P})$ is obtained.
\end{proof}

\section{Numerical Example\label{sec4}}

In this Section we shall consider an example, which appears in SAS Institute
Inc. (2013, Chapter 95) as well as in An (2002), in order to illustrate how
does the pseudo minimum phi-divergence estimator work for the multinomial
logistic regression with complex sample survey.%

\begin{table}[htbp]  \tabcolsep2.8pt  \centering
$%
\begin{tabular}
[c]{cc}\hline
\textbf{Class} & \textbf{Enrollment}\\\hline
\multicolumn{1}{l}{Freshman} & 3734\\
\multicolumn{1}{l}{Sophomore} & 3565\\
\multicolumn{1}{l}{Junior} & 3903\\
\multicolumn{1}{l}{Senior} & 4196\\\hline
\end{tabular}
$%
\caption{Number of student in each class of the target population for the survey.\label{table1}}%
\end{table}%

A market research firm conducts a survey among undergraduate students at the
University of North Carolina (UNC), at Chapel Hill, to evaluate three new web
designs at a commercial web-site targeting undergraduate students. The total
number of student in each class in the Fall semester of 2001 is shown in Table
\ref{table1}. The sample design is a stratified sample with clusters nested on
them, with the strata being the four students' classes and the clusters the
three web designs. Initially $100$ students were planned to be randomly
selected in each of the $n=12$ web designs using sample random sampling
(without replacement). For this reason, the weights for estimation are
considered to be $w_{1}=\frac{3734}{300}$, $w_{2}=\frac{3565}{300}$,
$w_{3}=\frac{3903}{300}$, $w_{4}=\frac{4196}{300}$. Since $m_{hi}=100$ for
$h=1,2,3,4=H$ (strata), $i=1,2,3=n_{h}$ (clusters) except for $m_{12}=90$ and
$m_{43}=97$, in practice some observations are missing values. Each student
selected in the sample is asked to evaluate the three Web designs and to rate
them ranging from dislike very much to like very much: ($1$) dislike very
much, ($2$) dislike, ($3$) neutral, ($4$) like, ($5=d+1$) like very much. The
survey results are collected and shown in Table \ref{table2}, with the three
different Web designs coded A, B and C. This table matches the one given in An
(2002) and the version appeared in SAS Institute Inc. (2013, Chapter 95) is
slightly different.%

\begin{table}[htbp]  \tabcolsep2.8pt  \centering
$%
\begin{tabular}
[c]{lcccccc}\hline
&  & \multicolumn{5}{c}{\textbf{Rating Counts}}\\\cline{3-7}%
\textbf{Strata} & \textbf{Design} & 1 & 2 & 3 & 4 & 5\\\hline
Freshman & A & \multicolumn{1}{r}{10} & \multicolumn{1}{r}{34} &
\multicolumn{1}{r}{25} & \multicolumn{1}{r}{16} & \multicolumn{1}{r}{15}\\
& B & \multicolumn{1}{r}{5} & \multicolumn{1}{r}{10} & \multicolumn{1}{r}{24}
& \multicolumn{1}{r}{30} & \multicolumn{1}{r}{21}\\
& C & \multicolumn{1}{r}{11} & \multicolumn{1}{r}{14} & \multicolumn{1}{r}{20}
& \multicolumn{1}{r}{34} & \multicolumn{1}{r}{21}\\\hline
Sophomore & A & \multicolumn{1}{r}{19} & \multicolumn{1}{r}{12} &
\multicolumn{1}{r}{26} & \multicolumn{1}{r}{18} & \multicolumn{1}{r}{25}\\
& B & \multicolumn{1}{r}{10} & \multicolumn{1}{r}{18} & \multicolumn{1}{r}{32}
& \multicolumn{1}{r}{23} & \multicolumn{1}{r}{17}\\
& C & \multicolumn{1}{r}{15} & \multicolumn{1}{r}{22} & \multicolumn{1}{r}{34}
& \multicolumn{1}{r}{9} & \multicolumn{1}{r}{20}\\\hline
Junior & A & \multicolumn{1}{r}{8} & \multicolumn{1}{r}{21} &
\multicolumn{1}{r}{23} & \multicolumn{1}{r}{26} & \multicolumn{1}{r}{22}\\
& B & \multicolumn{1}{r}{1} & \multicolumn{1}{r}{14} & \multicolumn{1}{r}{25}
& \multicolumn{1}{r}{23} & \multicolumn{1}{r}{37}\\
& C & \multicolumn{1}{r}{16} & \multicolumn{1}{r}{19} & \multicolumn{1}{r}{30}
& \multicolumn{1}{r}{23} & \multicolumn{1}{r}{12}\\\hline
Senior & A & \multicolumn{1}{r}{11} & \multicolumn{1}{r}{14} &
\multicolumn{1}{r}{24} & \multicolumn{1}{r}{33} & \multicolumn{1}{r}{18}\\
& B & \multicolumn{1}{r}{8} & \multicolumn{1}{r}{15} & \multicolumn{1}{r}{35}
& \multicolumn{1}{r}{30} & \multicolumn{1}{r}{12}\\
& C & \multicolumn{1}{r}{2} & \multicolumn{1}{r}{34} & \multicolumn{1}{r}{27}
& \multicolumn{1}{r}{18} & \multicolumn{1}{r}{16}\\\hline
\end{tabular}
\ \ \ \ \ $\caption{Evaluation of New Web Designs.\label{table2}}%
\end{table}%

The explanatory variables are qualitative, and valid to distinguish the
clusters within the strata. With respect to design A, it is given by
$\boldsymbol{x}_{h1}^{T}=\boldsymbol{x}_{1}^{T}=(1,0,0)$, $h=1,2,3,4$; with
respect to design B, by $\boldsymbol{x}_{h2}^{T}=\boldsymbol{x}_{2}%
^{T}=(0,1,0)$, $h=1,2,3,4$; with respect to design C, by $\boldsymbol{x}%
_{h3}^{T}=\boldsymbol{x}_{3}^{T}=(0,0,1)$, $h=1,2,3,4$. In Table \ref{table4}
every row represents the pseudo minimum Cressie-Read divergence estimates of
the $5$-dimensional probability vector $\boldsymbol{\pi}_{hi}%
(\widehat{\boldsymbol{\beta}}_{\phi_{\lambda},P})=\boldsymbol{\pi}%
_{i}(\widehat{\boldsymbol{\beta}}_{\phi_{\lambda},P})$, for the $i$-th cluster
$i=1,2,3$, for any stratum $h=1,2,3,4$, and a specific value in $\lambda
\in\{0,\frac{2}{3},1,1.5,2,2.5\}$. Each column of Table \ref{table3}
summarizes, first the pseudo minimum Cressie-Read divergence estimates of
$\boldsymbol{\beta}=(\boldsymbol{\beta}_{1}^{T},\boldsymbol{\beta}_{2}%
^{T},\boldsymbol{\beta}_{3}^{T},\boldsymbol{\beta}_{4}^{T})^{T}$, with
$\boldsymbol{\beta}_{i}^{T}=(\beta_{i1},\beta_{i2},\beta_{i3})$ $i=1,2,3,4$%
\ and\ $\lambda\in\{0,\frac{2}{3},1,1.5,2,2.5\}$, as well as the two versions
of the intra-cluster correlation estimates according to Theorems \ref{Th3} and
\ref{Th4} for the strata with the same cluster sizes, i.e. Sophomore ($2$) and
Junior ($3$). Section \ref{sec5} is devoted to study through simulation the
best choice for the value of $\lambda$ according to the root of the minimum
square error of $\widehat{\boldsymbol{\beta}}_{\phi_{\lambda},P}$,
$\widehat{\rho}^{2}(\widehat{\boldsymbol{\beta}}_{\phi_{\lambda},P})$ and
$\widetilde{\rho}^{2}(\widehat{\boldsymbol{\beta}}_{\phi_{\lambda},P})$.%

\begin{table}[htbp]  \tabcolsep2.8pt  \centering
$%
\begin{tabular}
[c]{ccccccc}\hline
&  & \multicolumn{5}{c}{\textbf{Rating Counts}}\\\cline{3-7}%
$\lambda$ & \textbf{Design} & 1 & 2 & 3 & 4 & 5\\\hline
$0$ & A & $0.1185$ & $0.2016$ & $0.2445$ & $0.2363$ & $0.1991$\\
& B & $0.0611$ & $0.1458$ & $0.2983$ & $0.2727$ & $0.2222$\\
& C & $0.1083$ & $0.2276$ & $0.2791$ & $0.2124$ & $0.1727$\\
$\frac{2}{3}$ & A & $0.1200$ & $0.2079$ & $0.2387$ & $0.2369$ & $0.1965$\\
& B & $0.0660$ & $0.1439$ & $0.2931$ & $0.2672$ & $0.2297$\\
& C & $0.1145$ & $0.2275$ & $0.2723$ & $0.2167$ & $0.1690$\\
$1$ & A & $0.1208$ & $0.2109$ & $0.2359$ & $0.2371$ & $0.1952$\\
& B & $0.0676$ & $0.1431$ & $0.2909$ & $0.2648$ & $0.2336$\\
& C & $0.1163$ & $0.2279$ & $0.2695$ & $0.2188$ & $0.1675$\\
$1.5$ & A & $0.1221$ & $0.2152$ & $0.2319$ & $0.2374$ & $0.1934$\\
& B & $0.0693$ & $0.1420$ & $0.2879$ & $0.2616$ & $0.2392$\\
& C & $0.1179$ & $0.2289$ & $0.2659$ & $0.2215$ & $0.1657$\\
$2$ & A & $0.1234$ & $0.2191$ & $0.2282$ & $0.2376$ & $0.1917$\\
& B & $0.0705$ & $0.1410$ & $0.2854$ & $0.2587$ & $0.2444$\\
& C & $0.1188$ & $0.2301$ & $0.2630$ & $0.2240$ & $0.1641$\\
$2.5$ & A & $0.1246$ & $0.2226$ & $0.2248$ & $0.2377$ & $0.1902$\\
& B & $0.0714$ & $0.1402$ & $0.2831$ & $0.2562$ & $0.2491$\\
& C & $0.1192$ & $0.2314$ & $0.2604$ & $0.2262$ & $0.1628$\\\hline
\end{tabular}
\ $%
\caption{Pseudo minimum Cressie-Read divergence estimates of probabilities for any of the four strata.\label{table4}}%
\end{table}%
%

\begin{table}[htbp]  \tabcolsep2.8pt  \centering
$%
\begin{tabular}
[c]{cccccccc}\hline
& \multicolumn{7}{c}{$\lambda$}\\\cline{2-8}
& $\qquad0$ & $\qquad\frac{2}{3}$ & $\qquad1$ & $\qquad1.5$ & $\qquad2$ &
$\qquad2.5$ & $\quad$\\\hline
$\widehat{\beta}_{11,\phi_{\lambda},P}$ & \multicolumn{1}{r}{$-0.5188$} &
\multicolumn{1}{r}{$-0.4933$} & \multicolumn{1}{r}{$-0.4802$} &
\multicolumn{1}{r}{$-0.4604$} & \multicolumn{1}{r}{$-0.4411$} &
\multicolumn{1}{r}{$-0.4228$} & \\
$\widehat{\beta}_{12,\phi_{\lambda},P}$ & \multicolumn{1}{r}{$-1.2910$} &
\multicolumn{1}{r}{$-1.2475$} & \multicolumn{1}{r}{$-1.2400$} &
\multicolumn{1}{r}{$-1.2381$} & \multicolumn{1}{r}{$-1.2424$} &
\multicolumn{1}{r}{$-1.2494$} & \\
$\widehat{\beta}_{13,\phi_{\lambda},P}$ & \multicolumn{1}{r}{$-0.4665$} &
\multicolumn{1}{r}{$-0.3889$} & \multicolumn{1}{r}{$-0.3649$} &
\multicolumn{1}{r}{$-0.3397$} & \multicolumn{1}{r}{$-0.3230$} &
\multicolumn{1}{r}{$-0.3116$} & \\\cline{2-8}%
$\widehat{\beta}_{21,\phi_{\lambda},P}$ & \multicolumn{1}{r}{$0.0127$} &
\multicolumn{1}{r}{$0.0564$} & \multicolumn{1}{r}{$0.0773$} &
\multicolumn{1}{r}{$0.1069$} & \multicolumn{1}{r}{$0.1336$} &
\multicolumn{1}{r}{$0.1573$} & \\
$\widehat{\beta}_{22,\phi_{\lambda},P}$ & \multicolumn{1}{r}{$-0.4210$} &
\multicolumn{1}{r}{$-0.4676$} & \multicolumn{1}{r}{$-0.4899$} &
\multicolumn{1}{r}{$-0.5213$} & \multicolumn{1}{r}{$-0.5498$} &
\multicolumn{1}{r}{$-0.5750$} & \\
$\widehat{\beta}_{23,\phi_{\lambda},P}$ & \multicolumn{1}{r}{$0.2761$} &
\multicolumn{1}{r}{$0.2974$} & \multicolumn{1}{r}{$0.3079$} &
\multicolumn{1}{r}{$0.3233$} & \multicolumn{1}{r}{$0.3380$} &
\multicolumn{1}{r}{$0.3517$} & \\\cline{2-8}%
$\widehat{\beta}_{31,\phi_{\lambda},P}$ & \multicolumn{1}{r}{$0.2056$} &
\multicolumn{1}{r}{$0.1947$} & \multicolumn{1}{r}{$0.1894$} &
\multicolumn{1}{r}{$0.1816$} & \multicolumn{1}{r}{$0.1741$} &
\multicolumn{1}{r}{$0.1670$} & \\
$\widehat{\beta}_{32,\phi_{\lambda},P}$ & \multicolumn{1}{r}{$0.2946$} &
\multicolumn{1}{r}{$0.2438$} & \multicolumn{1}{r}{$0.2196$} &
\multicolumn{1}{r}{$0.1857$} & \multicolumn{1}{r}{$0.1551$} &
\multicolumn{1}{r}{$0.1280$} & \\
$\widehat{\beta}_{33,\phi_{\lambda},P}$ & \multicolumn{1}{r}{$0.4803$} &
\multicolumn{1}{r}{$0.4770$} & \multicolumn{1}{r}{$0.4754$} &
\multicolumn{1}{r}{$0.4733$} & \multicolumn{1}{r}{$0.4714$} &
\multicolumn{1}{r}{$0.4697$} & \\\cline{2-8}%
$\widehat{\beta}_{41,\phi_{\lambda},P}$ & \multicolumn{1}{r}{$0.1715$} &
\multicolumn{1}{r}{$0.1870$} & \multicolumn{1}{r}{$0.1944$} &
\multicolumn{1}{r}{$0.2048$} & \multicolumn{1}{r}{$0.2143$} &
\multicolumn{1}{r}{$0.2228$} & \\
$\widehat{\beta}_{42,\phi_{\lambda},P}$ & \multicolumn{1}{r}{$0.2048$} &
\multicolumn{1}{r}{$0.1512$} & \multicolumn{1}{r}{$0.1256$} &
\multicolumn{1}{r}{$0.0896$} & \multicolumn{1}{r}{$0.0570$} &
\multicolumn{1}{r}{$0.0280$} & \\
$\widehat{\beta}_{43,\phi_{\lambda},P}$ & \multicolumn{1}{r}{$0.2070$} &
\multicolumn{1}{r}{$0.2488$} & \multicolumn{1}{r}{$0.2668$} &
\multicolumn{1}{r}{$0.2906$} & \multicolumn{1}{r}{$0.3111$} &
\multicolumn{1}{r}{$0.3288$} & \\\hline
$\widehat{\rho}_{2}^{2}(\widehat{\boldsymbol{\beta}}_{\phi_{\lambda},P})$ &
\multicolumn{1}{r}{$0.0119$} & \multicolumn{1}{r}{$0.0123$} &
\multicolumn{1}{r}{$0.0127$} & \multicolumn{1}{r}{$0.0135$} &
\multicolumn{1}{r}{$0.0142$} & \multicolumn{1}{r}{$0.0150$} & \\
$\widetilde{\rho}_{2}^{2}(\widehat{\boldsymbol{\beta}}_{\phi_{\lambda},P})$ &
\multicolumn{1}{r}{$0.0119$} & \multicolumn{1}{r}{$0.0048$} &
\multicolumn{1}{r}{$0.0051$} & \multicolumn{1}{r}{$0.0056$} &
\multicolumn{1}{r}{$0.0061$} & \multicolumn{1}{r}{$0.0067$} & \\\cline{2-8}%
$\widehat{\rho}_{3}^{2}(\widehat{\boldsymbol{\beta}}_{\phi_{\lambda},P})$ &
\multicolumn{1}{r}{$0.0088$} & \multicolumn{1}{r}{$0.0072$} &
\multicolumn{1}{r}{$0.0066$} & \multicolumn{1}{r}{$0.0059$} &
\multicolumn{1}{r}{$0.0054$} & \multicolumn{1}{r}{$0.0051$} & \\
$\widetilde{\rho}_{3}^{2}(\widehat{\boldsymbol{\beta}}_{\phi_{\lambda},P})$ &
\multicolumn{1}{r}{$0.0088$} & \multicolumn{1}{r}{$0.0014$} &
\multicolumn{1}{r}{$0.0010$} & \multicolumn{1}{r}{$0.0006$} &
\multicolumn{1}{r}{$0.0003$} & \multicolumn{1}{r}{$0.0000$} & \\\hline
\end{tabular}
\ \ $%
\caption{Pseudo minimum Cressie-Read divergence estimates of $\boldsymbol{\beta }$ and $\rho ^{2}$.\label{table3}}%
\end{table}%
\pagebreak

\section{Simulation Study\label{sec5}}

In order to analyze the performance of the proposed estimators through root of
the mean square errors (RMSE), an adapted design focussed in the simulation
experiment proposed in Morel (1989) is conducted. Based on a unique stratum
with $n$ clusters of the same size $m$, three overdispersed multinomial
distributions for $\widehat{\boldsymbol{Y}}_{i}$\ described as%
\begin{align*}
\boldsymbol{E}[\widehat{\boldsymbol{Y}}_{i}]  &  =m\boldsymbol{\pi}_{i}\left(
\boldsymbol{\beta}_{0}\right)  \quad\text{and}\quad\boldsymbol{V}%
[\widehat{\boldsymbol{Y}}_{i}]=\nu_{m}m\boldsymbol{\Delta}(\boldsymbol{\pi
}_{i}\left(  \boldsymbol{\beta}_{0}\right)  ),\\
\nu_{m}  &  =1+\rho^{2}(m-1),
\end{align*}
are considered for $i=1,...,n$, the Dirichlet-multinomial (DM), the
random-clumped (RC) and the $m$-inflated distribution ($m$-I), all of them
with the same parameters $\boldsymbol{\pi}_{i}\left(  \boldsymbol{\beta}%
_{0}\right)  $ and $\rho$\ (see Appendix of Alonso et al. (2016) for details
of their generators). The value of the true probability associated with the
$i$-th cluster is $\boldsymbol{\pi}_{i}\left(  \boldsymbol{\beta}_{0}\right)
=(\pi_{i1}\left(  \boldsymbol{\beta}_{0}\right)  ,\pi_{i2}\left(
\boldsymbol{\beta}_{0}\right)  ,\pi_{i3}\left(  \boldsymbol{\beta}_{0}\right)
,\pi_{i4}\left(  \boldsymbol{\beta}_{0}\right)  )^{T}$, where
\[
\boldsymbol{\pi}_{ir}\left(  \boldsymbol{\beta}_{0}\right)  =\dfrac
{\exp\{\boldsymbol{x}_{i}^{T}\boldsymbol{\beta}_{r,0}\}}{%
{\textstyle\sum_{s=1}^{d+1}}
\exp\{\boldsymbol{x}_{i}^{T}\boldsymbol{\beta}_{s,0}\}},\quad r=1,2,3,4,
\]
$\boldsymbol{\beta}=(\boldsymbol{\beta}_{1}^{T},\boldsymbol{\beta}_{2}%
^{T},\boldsymbol{\beta}_{3}^{T},\boldsymbol{\beta}_{4}^{T})^{T}$, with
$\boldsymbol{\beta}_{1}^{T}=(-0.3,-0.1,0.1,0.2)$, $\boldsymbol{\beta}_{2}%
^{T}=(0.2,-0.2,-0.2,0.1)$, $\boldsymbol{\beta}_{3}^{T}=(-0.1,0.3,-0.3,0.1)$,
$\boldsymbol{\beta}_{4}^{T}=(0,0,0,0)$, and%
\[
\boldsymbol{x}_{i}\overset{ind}{\sim}\mathcal{N}(\boldsymbol{\mu
},\boldsymbol{\Sigma}),\quad\boldsymbol{\mu}=(1,-2,1,5)^{T},\quad
\boldsymbol{\Sigma}=\mathrm{diag}\{0,25,25,25\},\quad i=1,\ldots,n,
\]
while the value true intra-cluster correlation parameter, $\rho^{2}$, is
different depending on the scenario. Notice that $d=3$ and $k=4$, and the
values of $n$ and $m$ are different depending on the scenario.

\begin{itemize}
\item Scenario 1: $n=60$, $m=21$, $\rho^{2}\in\{0.05i\}_{i=0}^{19}$, DM, RC
and $m$-I distributions (Figures \ref{fig1}-\ref{fig3});

\item Scenario 2: $n\in\{10i\}_{i=1}^{15}$, $m=21$, $\rho^{2}=0.25$, RC
distribution (Figure \ref{fig4});

\item Scenario 3: $n=60$, $m\in\{10i\}_{i=1}^{10}$, $\rho^{2}=0.25$, RC
distribution (Figures \ref{fig5}-\ref{fig6}, above);

\item Scenario 4: $n=60$, $m\in\{10i\}_{i=1}^{10}$, $\rho^{2}=0.75$, RC
distribution (Figures \ref{fig5}-\ref{fig6}, middle);

\item Scenario 5: $n=20$, $m\in\{10i\}_{i=1}^{10}$, $\rho^{2}=0.25$, RC
distribution (Figures \ref{fig5}-\ref{fig6}, below).
\end{itemize}

In the previous scenarios the RMSE for the pseudo minimum Cressie-Read
divergence estimators of $\boldsymbol{\beta}$ with $\lambda\in\{0,\frac{2}%
{3},1,1.5,2,2.5\}$ are studied, as well as for the estimators of $\rho^{2}$,
depending on the method (of moments or Binder) and the value of $\lambda$ to
estimate $\boldsymbol{\beta}$ (ordinal axis of the plots). As expected from a
theoretical point of view, the simulations show that the RMSE increases as
$\rho^{2}$ increases, $n$ decreases or $m$ decreases.

For $\boldsymbol{\beta}$, the interest of the pseudo minimum Cressie-Read
divergence estimators is clearly justified for small-moderate sizes of $n$ and
strong-moderate intra-cluster correlation. The cluster size, $m$, affects but
not so much as the number of clusters, $n$. More thoroughly, in these cases,
the value of $\lambda\in\{\frac{2}{3},1,1.5,2,2.5\}$ exhibits better
performance than the pseudo maximum likelihood estimator ($\lambda=0$).

For the estimators of the intra-cluster correlation coefficient two clear and
important findings, valid for any value of $n$, $m$, or true value of
$\rho^{2}$\ , are:

\begin{description}
\item[*] The estimator of $\rho^{2}$ with the method of of moments is not
recommended, since the estimator with the Binder's method is much better.

\item[*] The best estimator of $\rho^{2}$ with the Binder's method is obtained
with $\lambda=\frac{2}{3}$.
\end{description}%

\begin{figure}[htbp]  \centering
$\
\begin{tabular}
[c]{c}%
{\includegraphics[
trim=0.000000in 1.056453in 0.000000in 0.000000in,
height=2.7466in,
width=4.4624in
]%
{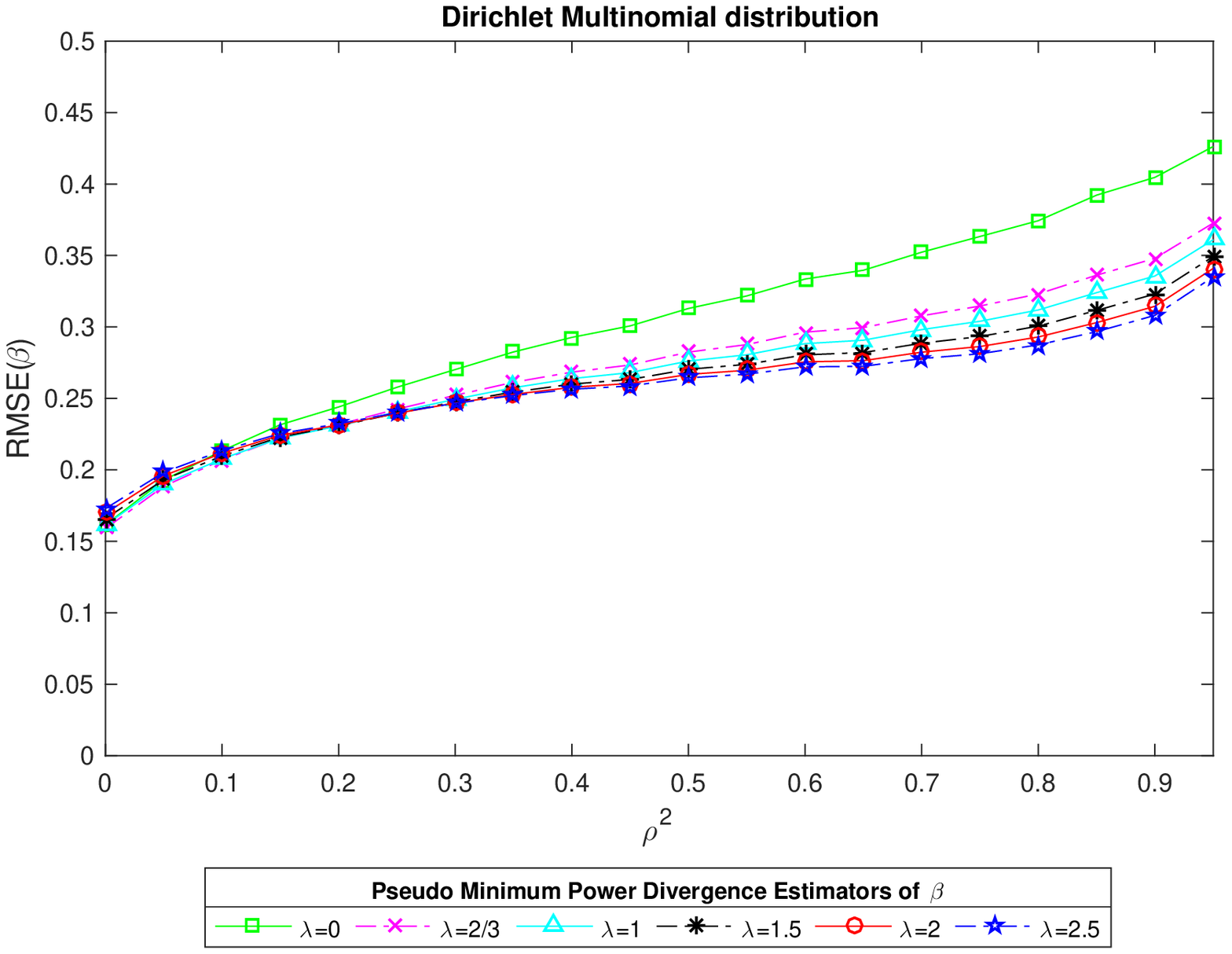}%
}
\\%
{\includegraphics[
trim=0.000000in 1.056453in 0.000000in 0.000000in,
height=2.7466in,
width=4.4624in
]%
{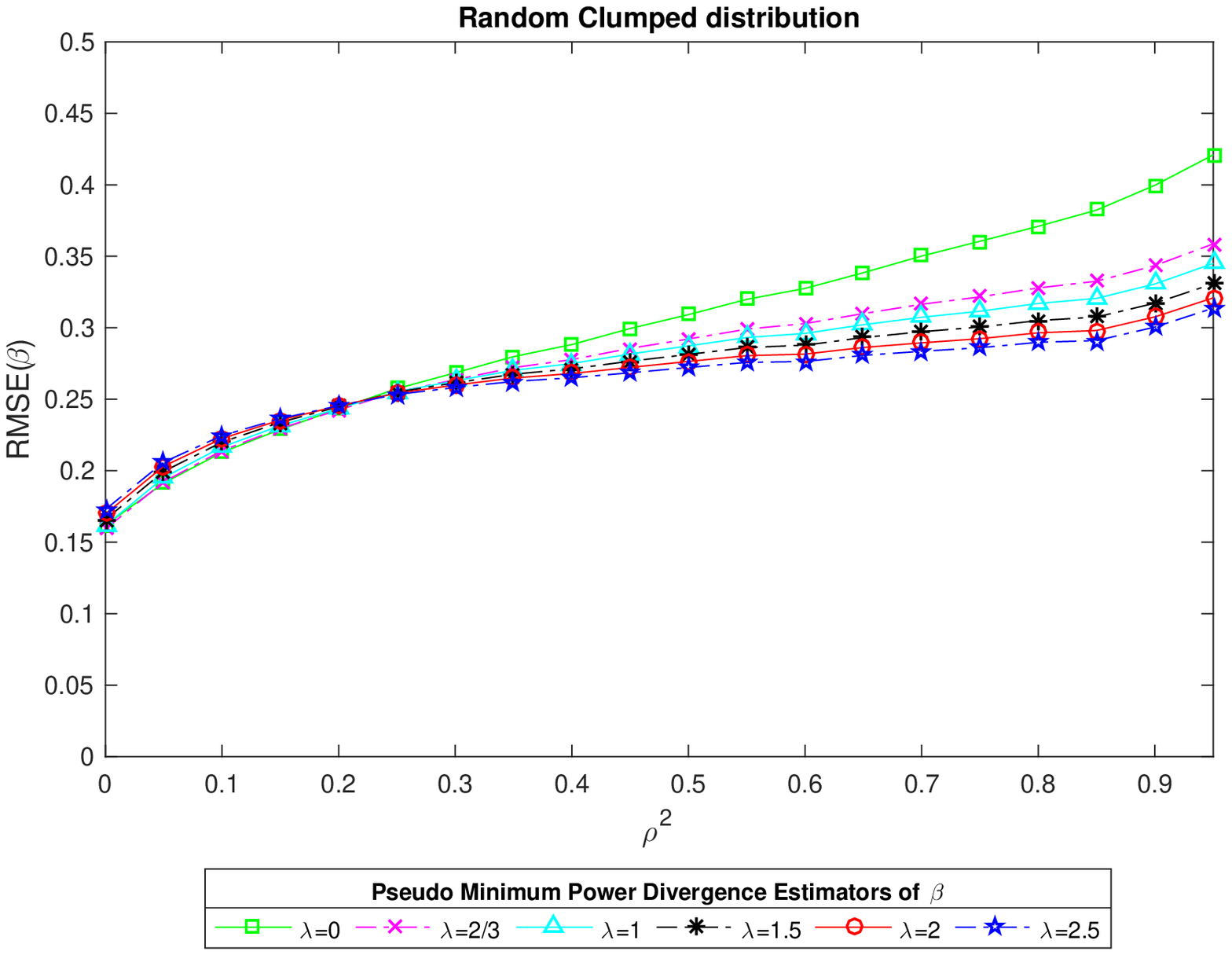}%
}
\\%
{\includegraphics[
height=3.3486in,
width=4.4624in
]%
{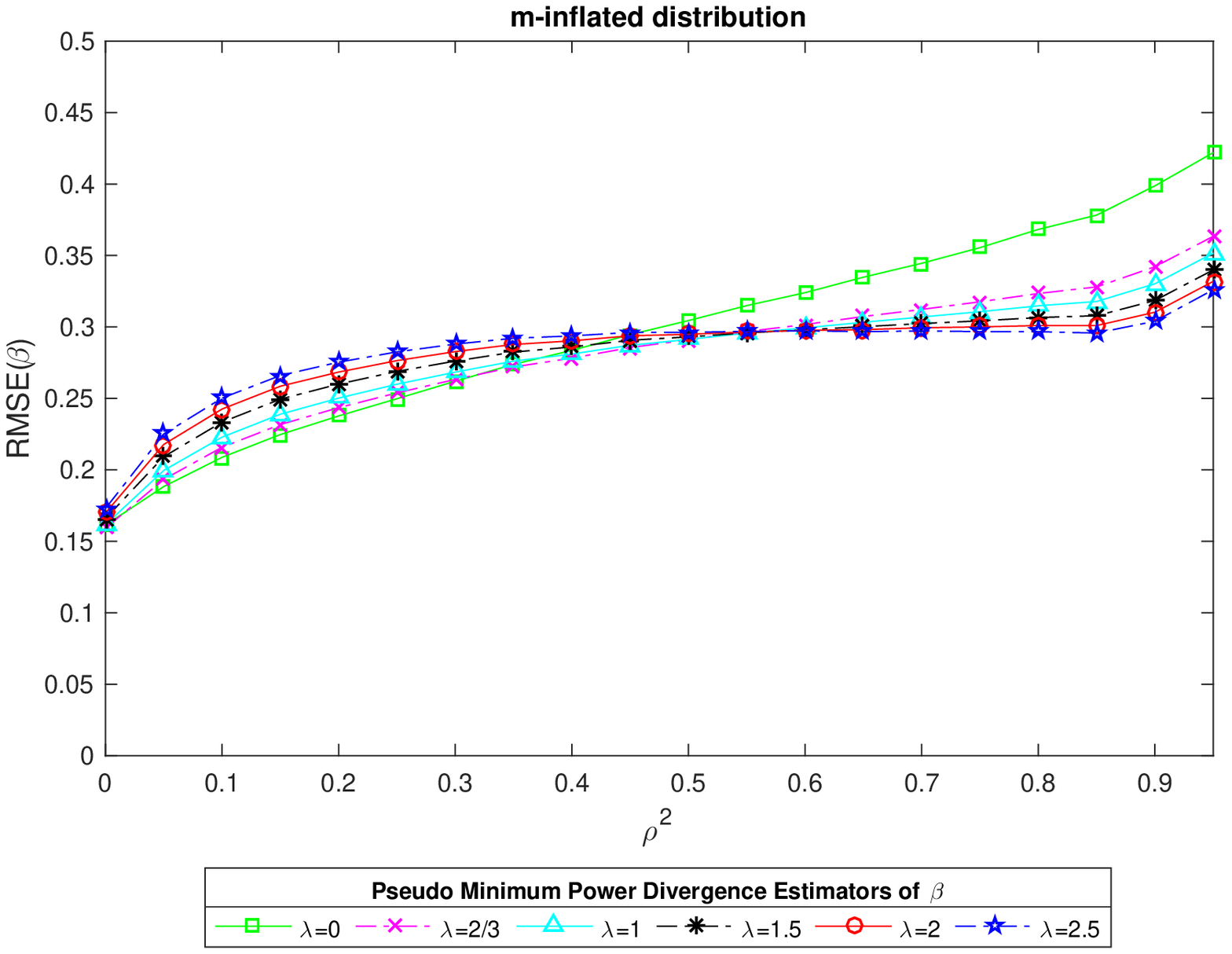}%
}
\end{tabular}
\ $%
\caption{RMSEs of of seudo minimum Cressie-Read divergence estimators of $\boldsymbol{\beta }$ for three distributions.\label{fig1}}%
\end{figure}%
%

\begin{figure}[htbp]  \centering
$\
\begin{tabular}
[c]{c}%
{\includegraphics[
trim=0.000000in 1.056453in 0.000000in 0.000000in,
height=2.7475in,
width=4.4624in
]%
{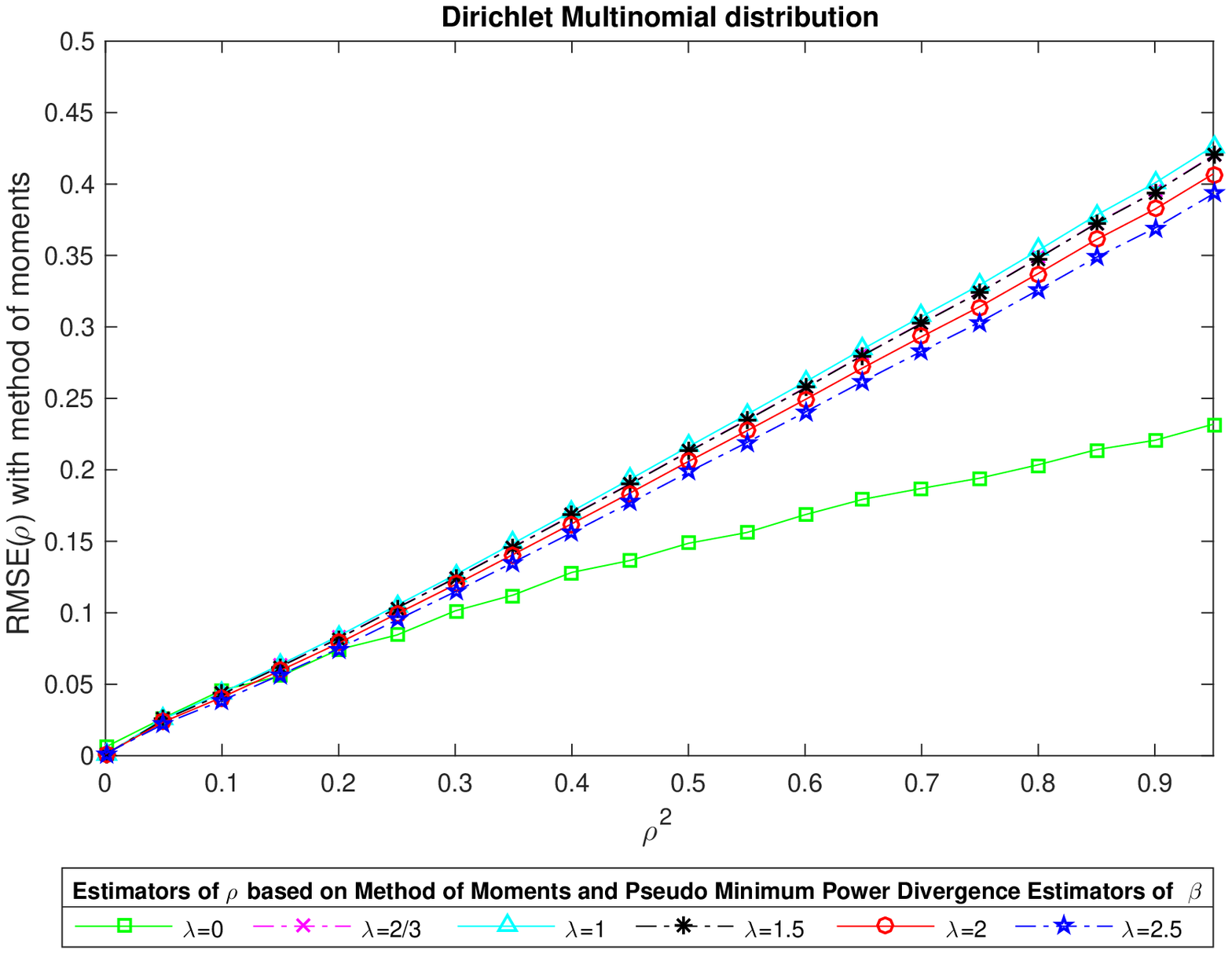}%
}
\\%
{\includegraphics[
trim=0.000000in 1.056453in 0.000000in 0.000000in,
height=2.7475in,
width=4.4624in
]%
{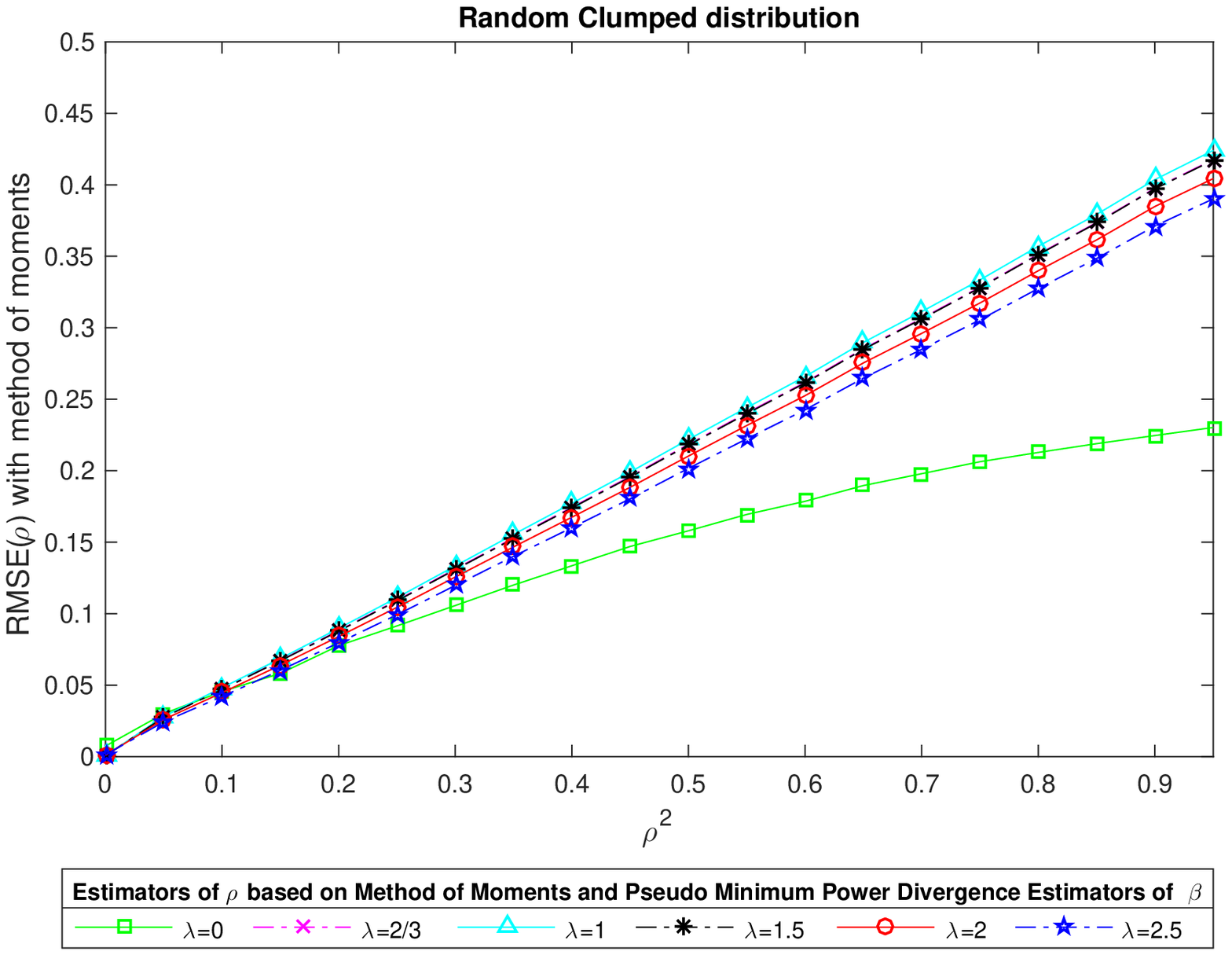}%
}
\\%
{\includegraphics[
height=3.3486in,
width=4.4624in
]%
{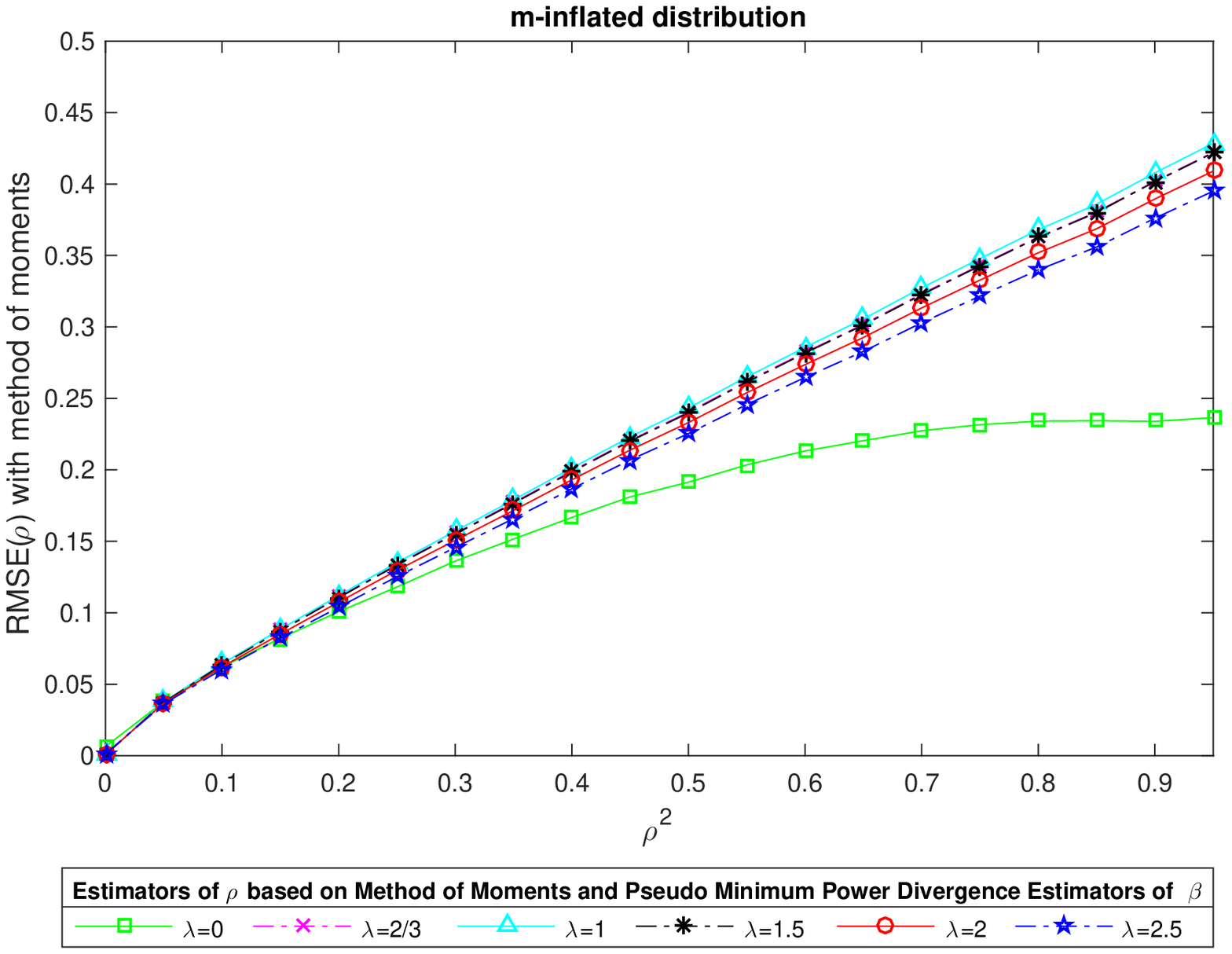}%
}
\end{tabular}
\ $%
\caption{RMSEs of estimators of $\rho ^{2}$ based on the method of moments for three distributions.\label{fig2}}%
\end{figure}%
%

\begin{figure}[htbp]  \centering
$\
\begin{tabular}
[c]{c}%
{\includegraphics[
trim=0.000000in 1.056453in 0.000000in 0.000000in,
height=2.7466in,
width=4.4624in
]%
{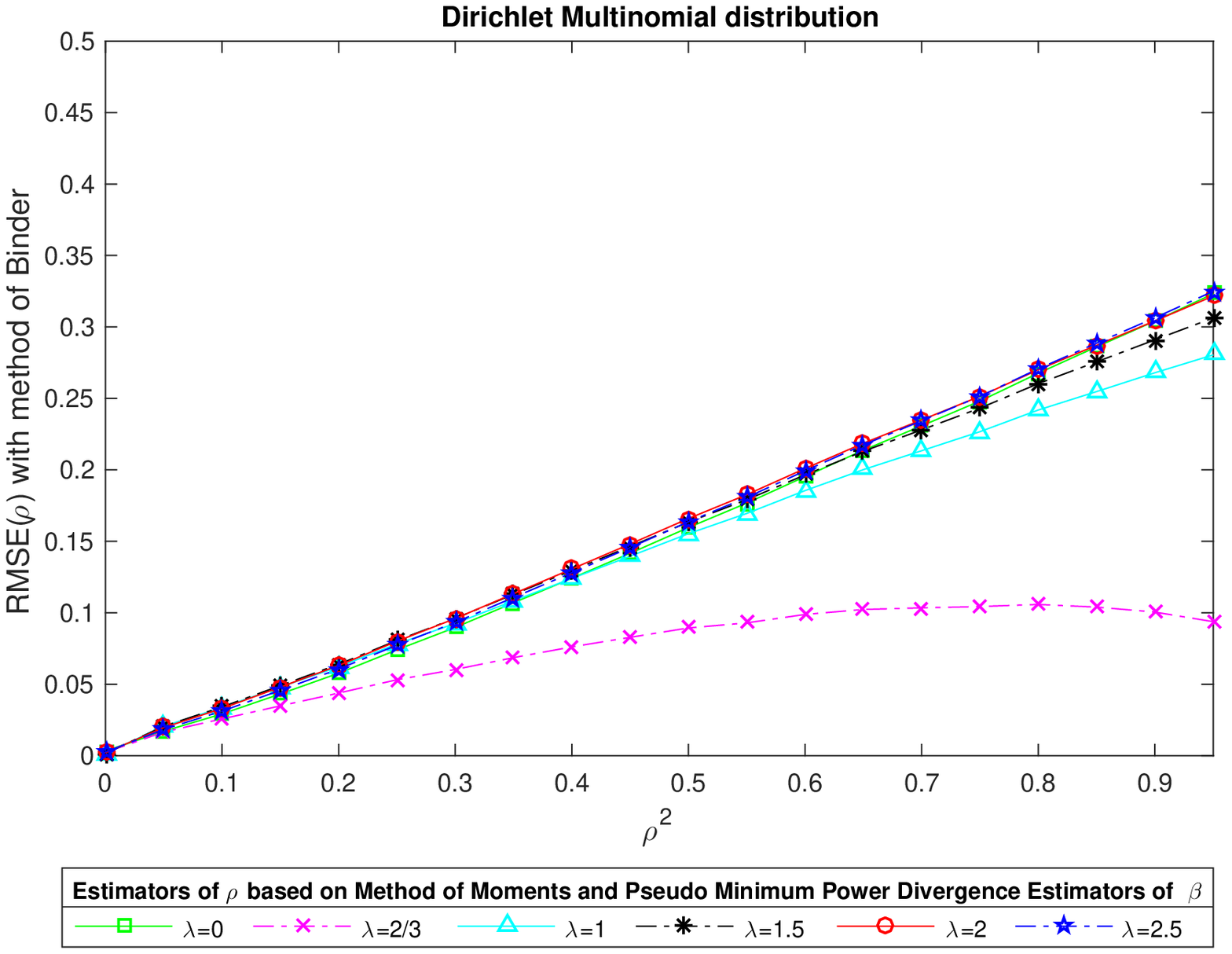}%
}
\\%
{\includegraphics[
trim=0.000000in 1.056453in 0.000000in 0.000000in,
height=2.7475in,
width=4.4616in
]%
{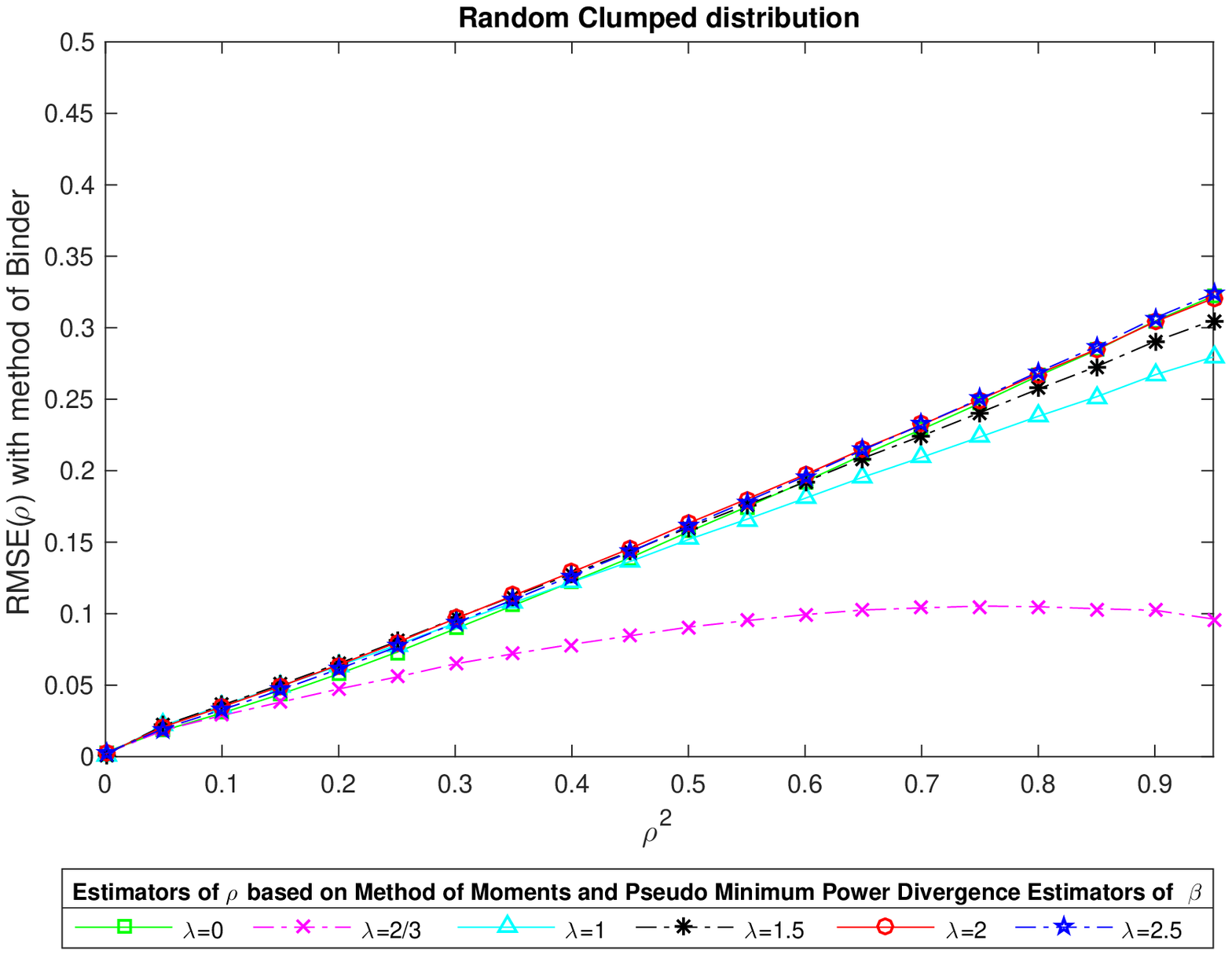}%
}
\\%
{\includegraphics[
height=3.3486in,
width=4.4624in
]%
{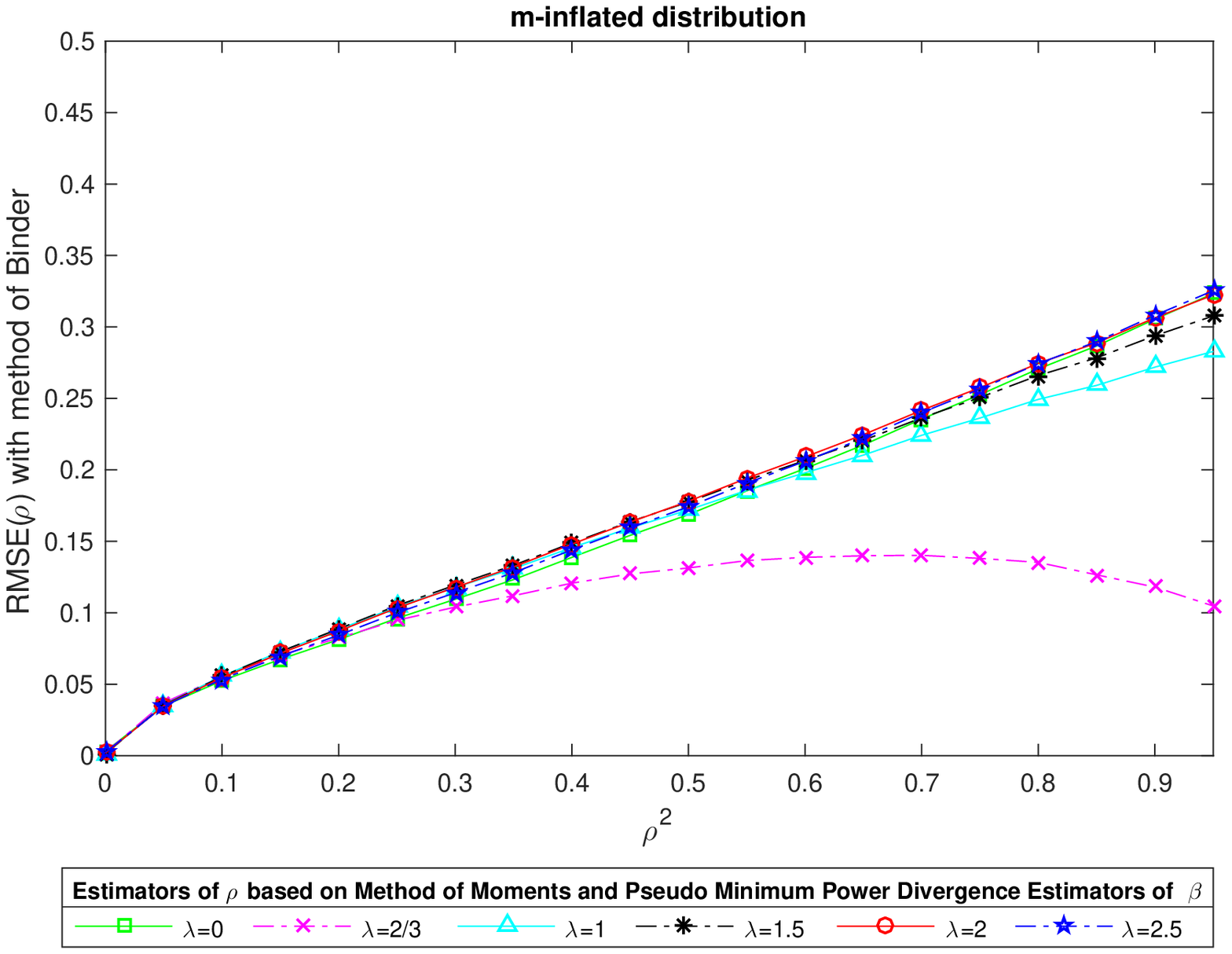}%
}
\end{tabular}
\ $%
\caption{RMSEs of estimators of $\rho ^{2}$ based on the method of Binder.\label{fig3}}%
\end{figure}%
%

\begin{figure}[htbp]  \centering
$\
\begin{tabular}
[c]{c}%
{\includegraphics[
trim=0.000000in 1.056453in 0.000000in 0.000000in,
height=2.7475in,
width=4.4624in
]%
{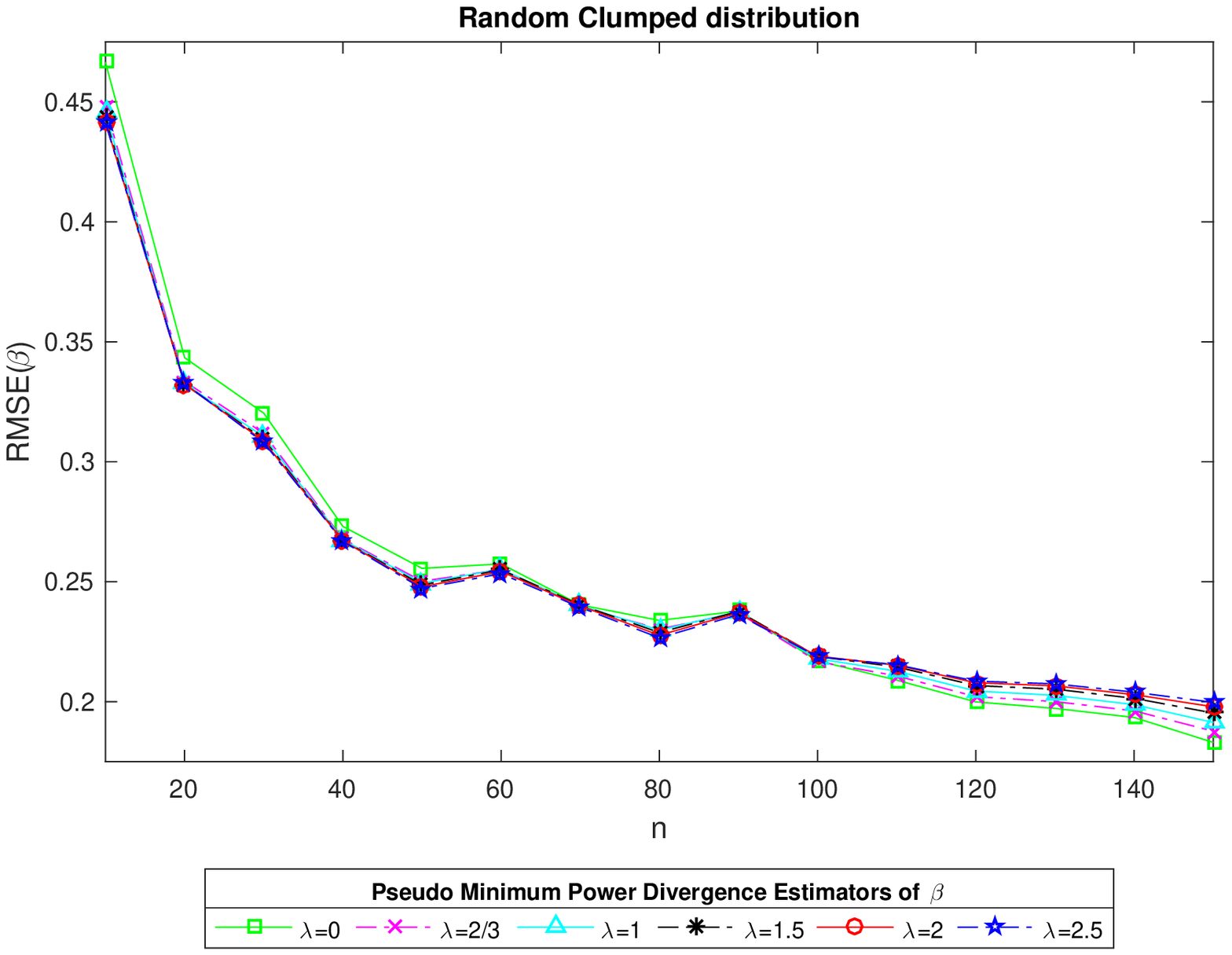}%
}
\\%
{\includegraphics[
trim=0.000000in 1.056453in 0.000000in 0.000000in,
height=2.7475in,
width=4.4624in
]%
{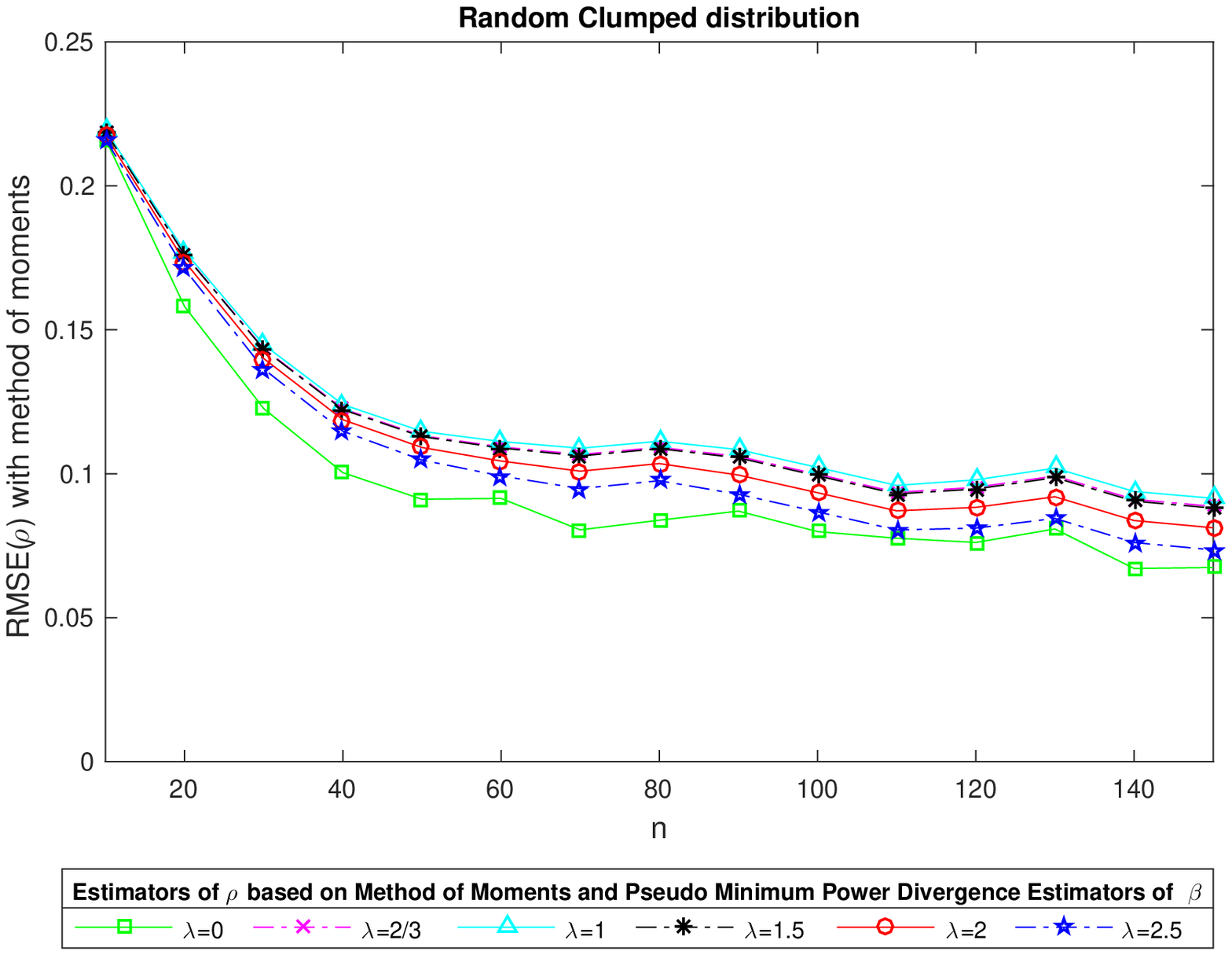}%
}
\\%
{\includegraphics[
height=3.3501in,
width=4.4624in
]%
{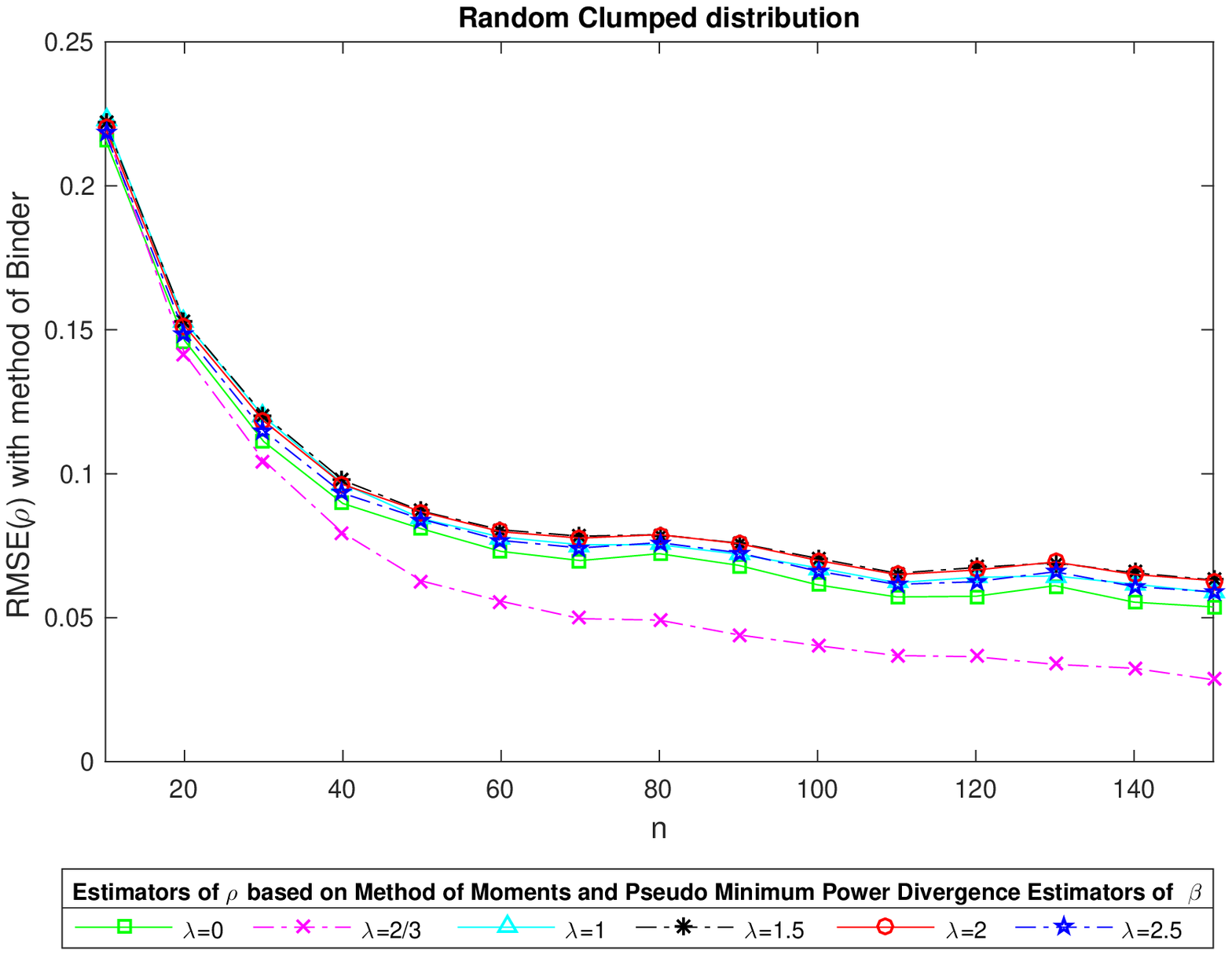}%
}
\end{tabular}
\ \ \ $%
\caption{RMSEs of estimators of $\boldsymbol{\beta }$ and $\rho ^{2}$ when the total number of clusters, $n$, increases, for the random clumped distribution. Case $m=21$, $\rho=0.25$.\label{fig4}}%
\end{figure}%
%

\begin{figure}[htbp]  \centering
$\
\begin{tabular}
[c]{c}%
{\includegraphics[
trim=0.000000in 1.056453in 0.000000in 0.000000in,
height=2.7475in,
width=4.4624in
]%
{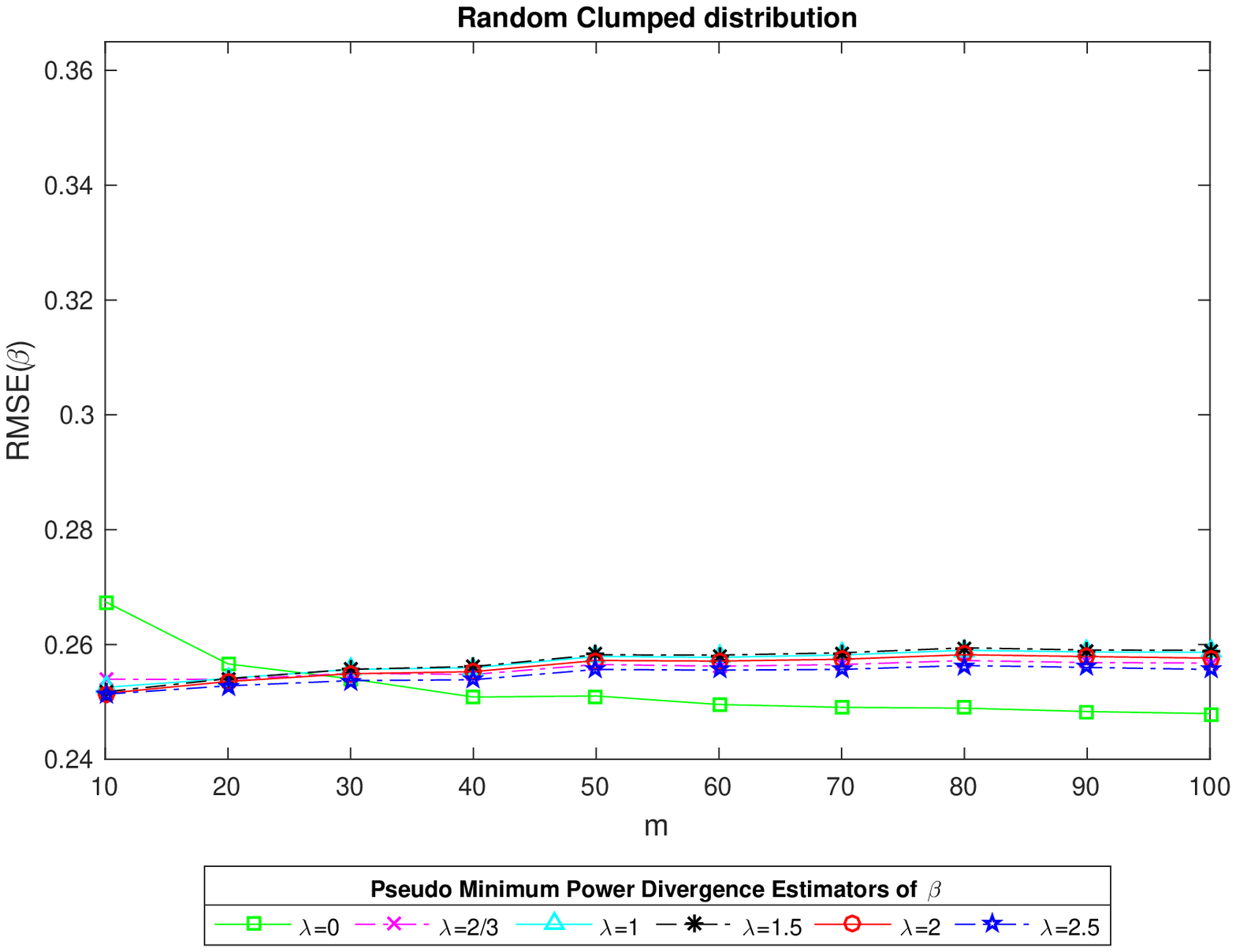}%
}
\\%
{\includegraphics[
trim=0.000000in 1.056453in 0.000000in 0.000000in,
height=2.7466in,
width=4.4624in
]%
{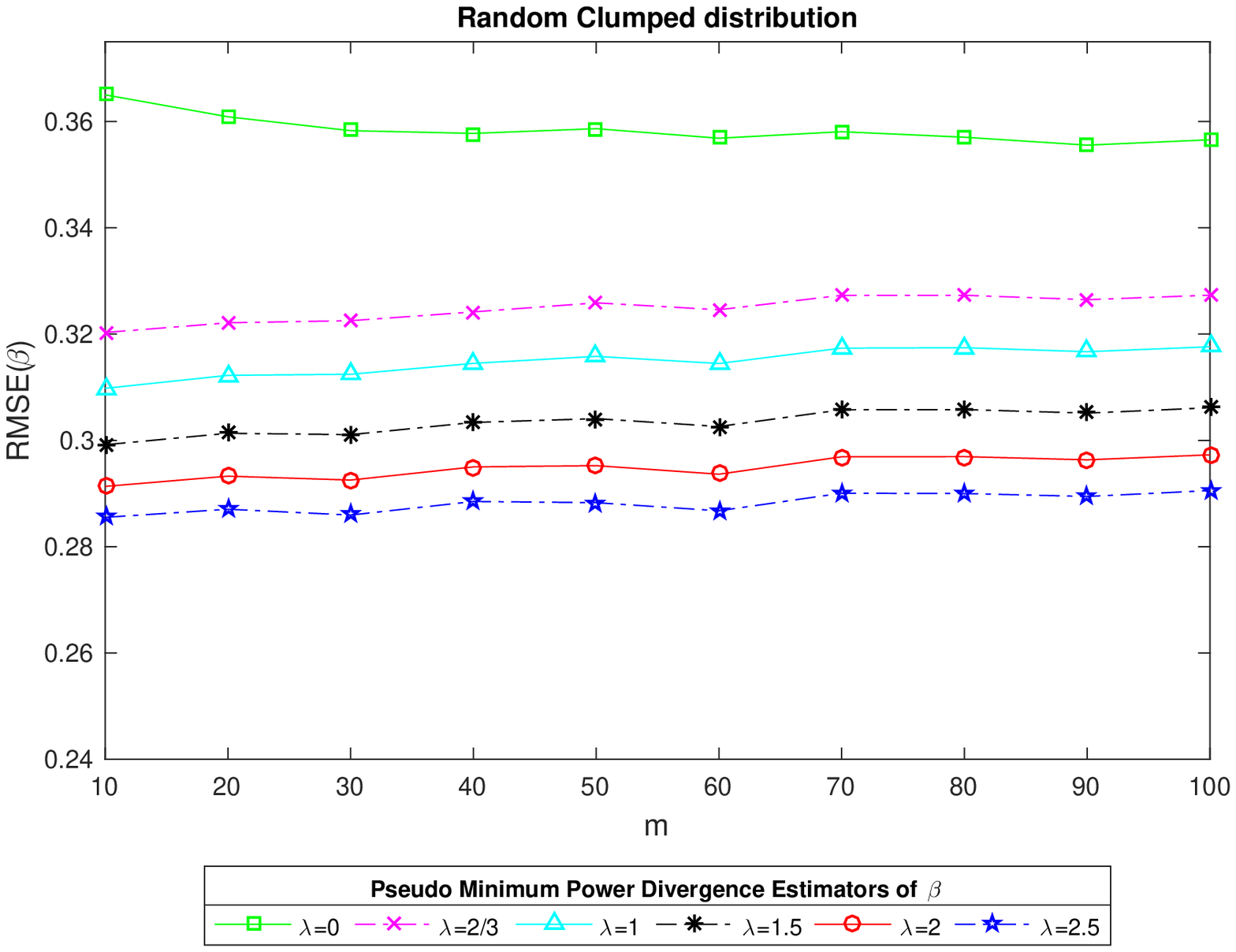}%
}
\\%
{\includegraphics[
height=3.3486in,
width=4.4624in
]%
{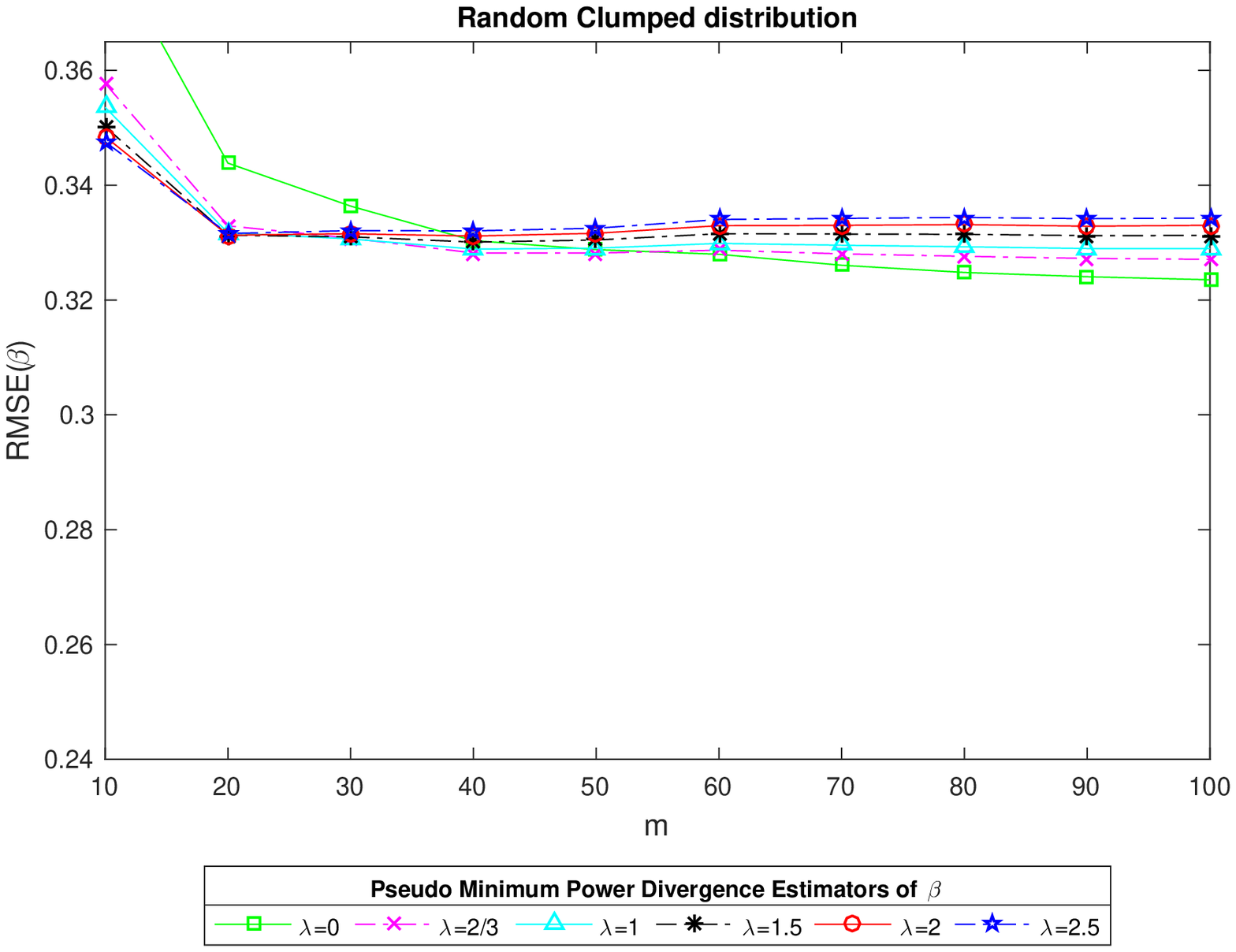}%
}
\end{tabular}
\ \ \ \ $%
\caption{RMSEs of estimators of $\boldsymbol{\beta }$ when the number of individuals within clusters, $m$, increases, for the random clumped distribution. Cases: $n=60$, $\rho=0.25$ (above), $n=60$, $\rho=0.75$ (middle), $n=20$, $\rho=0.25$ (below).\label{fig5}}%
\end{figure}%
%

\begin{figure}[htbp]  \centering
$\
\begin{tabular}
[c]{c}%
{\includegraphics[
trim=0.000000in 1.056453in 0.000000in 0.000000in,
height=2.7466in,
width=4.4624in
]%
{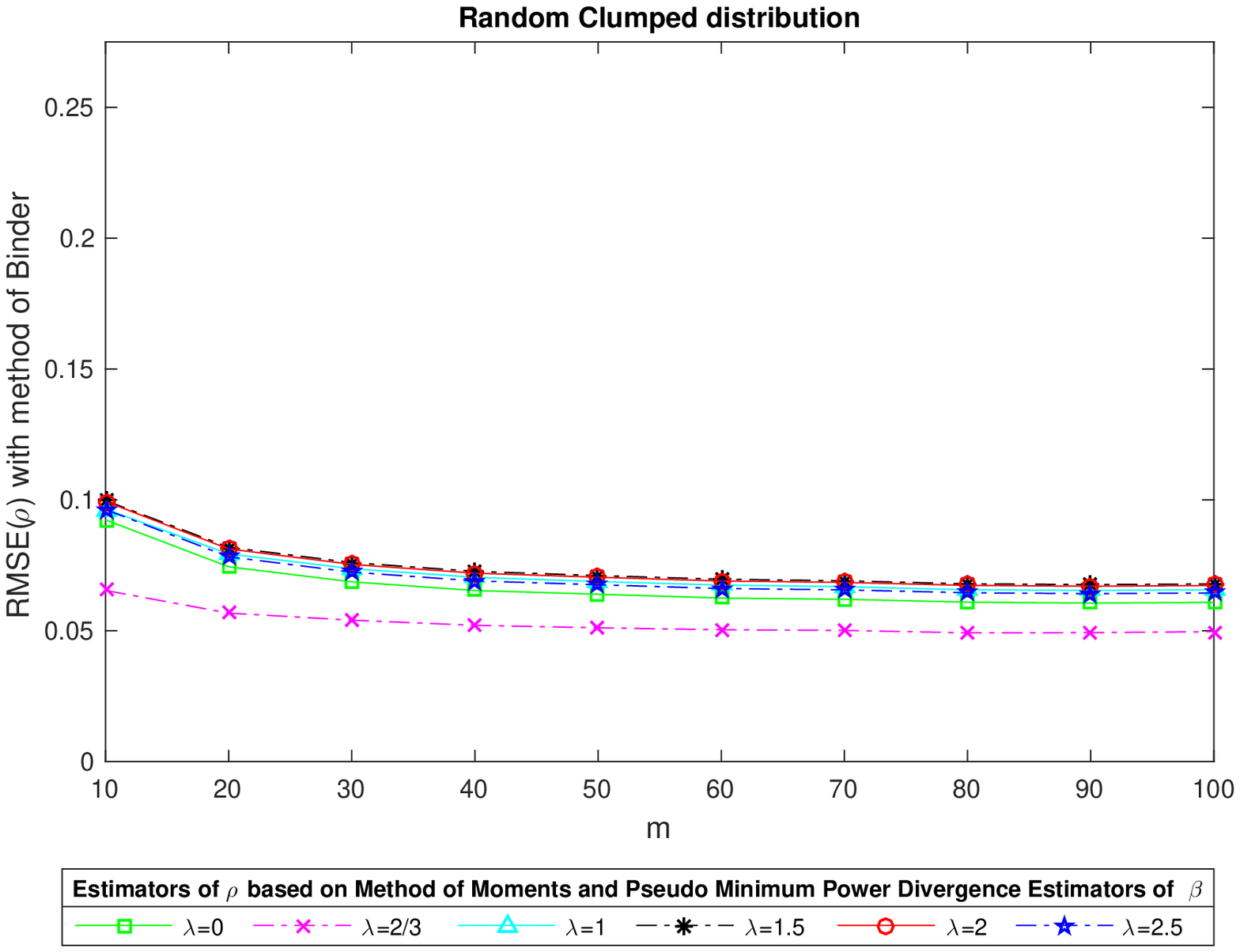}%
}
\\%
{\includegraphics[
trim=0.000000in 1.056453in 0.000000in 0.000000in,
height=2.7475in,
width=4.4624in
]%
{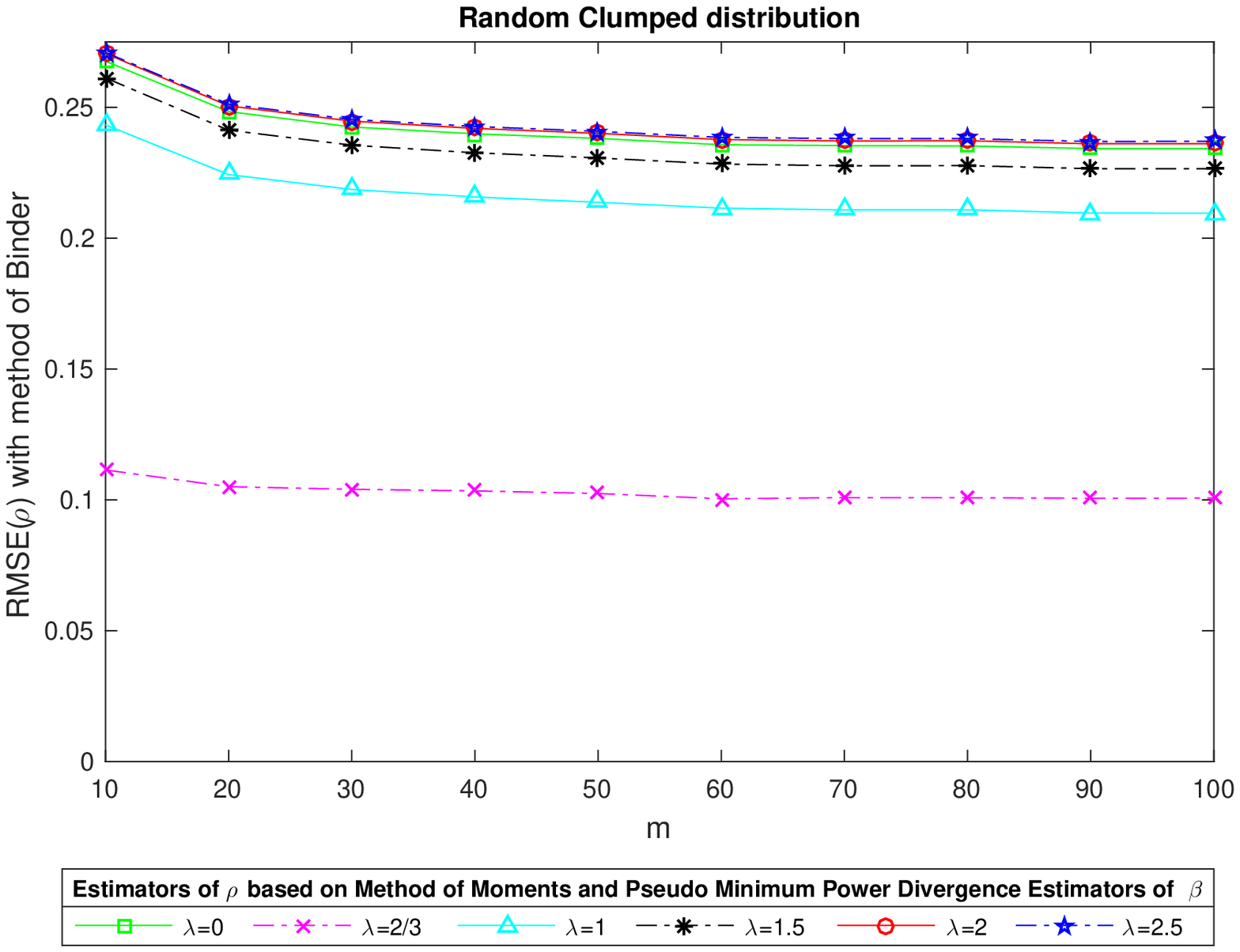}%
}
\\%
{\includegraphics[
height=3.3486in,
width=4.4624in
]%
{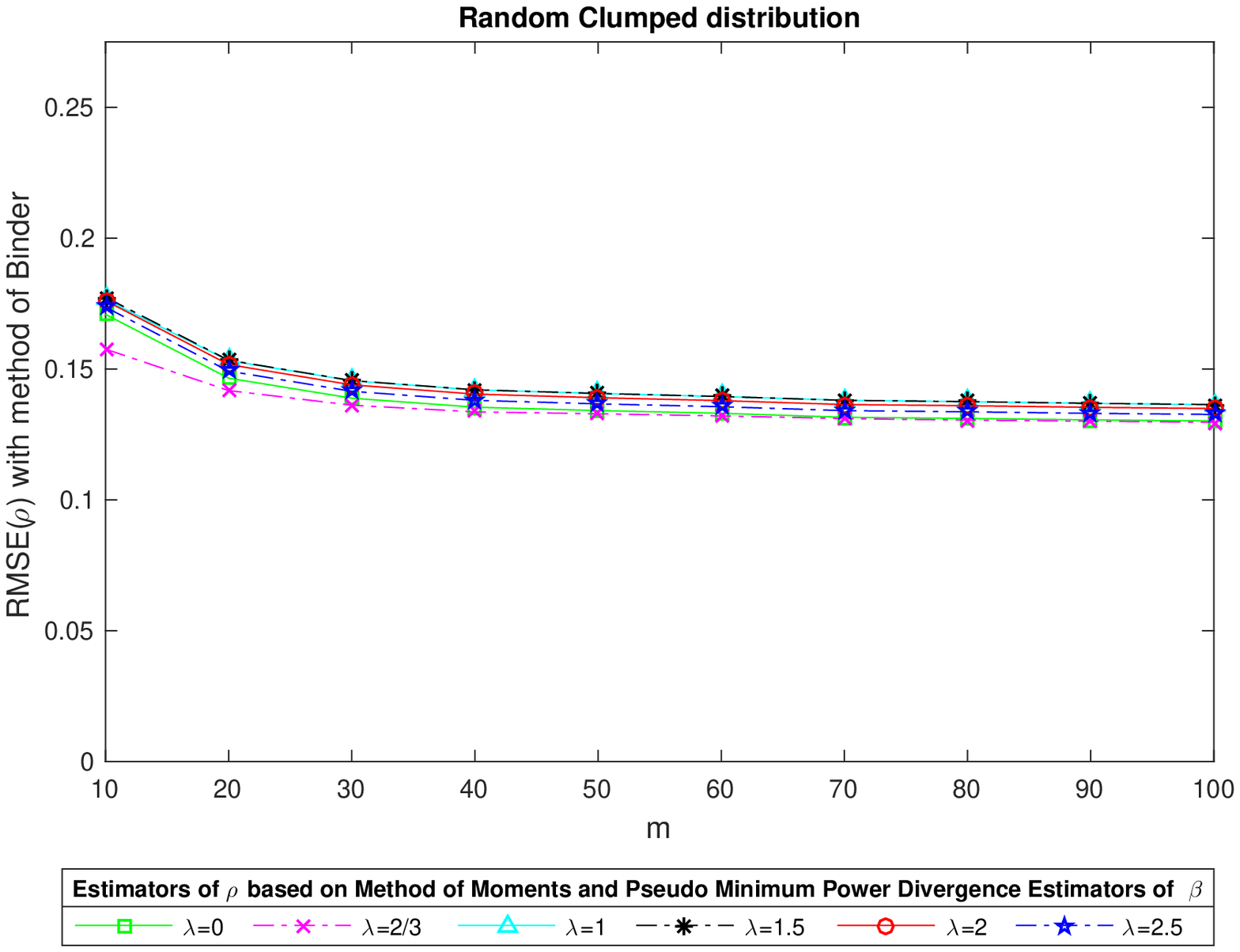}%
}
\end{tabular}
\ \ \ \ \ \ $%
\caption{RMSEs of estimators of $\rho ^{2}$ (Binder's method) when the number of individuals within clusters, $m$, increases, for the random clumped distribution. Cases: $n=60$, $\rho=0.25$ (above), $n=60$, $\rho=0.75$ (middle), $n=20$, $\rho=0.25$ (below).\label{fig6}}%
\end{figure}%

\section{Concluding remarks\label{sec6}}

Even though the multinomial logistic regression is an extensively applied
model, in our knowledge there is no study which compares the method of moments
and the Binder's method for estimating the intracluster correlation
coefficient. The simulation study designed in this paper shows that the
Binder's method is by far the best choice.

As future research, we would like to extend the proposed method to be valid
for estimating the $\boldsymbol{\beta}$ and $\rho^{2}$ for different cluster sizes.

\pagebreak

\end{document}